\documentclass[11pt]{article}
\usepackage[a4paper]{geometry}
\usepackage{color}
\usepackage{array}
\usepackage{multirow}
\usepackage{amsmath}
\usepackage{amssymb}
\usepackage{graphicx}
\usepackage[unicode=true,pdfusetitle,
 bookmarks=true,bookmarksnumbered=true,bookmarksopen=false,
 breaklinks=false,pdfborder={0 0 1},backref=section,colorlinks=true]
 {hyperref}
\hypersetup{
 pdftex,unicode}

\makeatletter

\providecommand{\tabularnewline}{\\}

\@ifundefined{date}{}{\date{}}

\usepackage{amsthm}
\usepackage{graphics}\usepackage{epsfig}

 \renewenvironment{proof}[1][\proofname]
{\par\pushQED{\qed}\normalfont\topsep6\p@\@plus6\p@\relax\trivlist\item[\hskip\labelsep\bfseries#1\@addpunct{.}]\ignorespaces}{\popQED\endtrivlist\@endpefalse} 

\theoremstyle{plain}
\newtheorem{thm}{Theorem}
\newtheorem{lem}{Lemma}
\newtheorem{defn}{Definition}
\newtheorem{exmp}{Example}

\newtheorem{cor}{Corollary}
\newtheorem{prop}{Proposition}

\date{August 31, 2015}

\makeatother

\begin{document}

\author{Yangbo Song\thanks{Department of Economics, UCLA. Email: darcy07@ucla.edu.}
\and Mihaela van der Schaar\thanks{Department of Electrical Engineering, UCLA, and Director of UCLA Center for Engineering Economics, Learning, and Networks. Email: mihaela@ee.ucla.edu.}}

\title{Dynamic Network Formation with Foresighted Agents\thanks{We are grateful to William Zame, Simpson Zhang, Kartik Ahuja and a number of seminar audiences for suggestions which have significantly improved the paper. This work was partially funded by ONR Grant N00014-15-1-2038.}}
\maketitle
\begin{abstract}
What networks can form and persist when agents are self-interested?
Can such networks be efficient? A substantial theoretical literature
predicts that the only networks that can form and persist must have
very special shapes and that such networks cannot be efficient, but
these predictions are in stark contrast to empirical findings. In
this paper, we present a new model of network formation. In contrast
to the existing literature, our model is dynamic (rather than static),
we model agents as \textit{foresighted} (rather than myopic) and we
allow for the possibility that agents are \textit{heterogeneous} (rather
than homogeneous). We show that a very wide variety of networks can
form and persist; in particular, efficient networks \textit{can} form
and persist if they provide every agent a strictly positive payoff.
For the widely-studied connections model, we provide a full characterization
of the set of efficient networks that can form and persist. Our predictions
are consistent with empirical findings.

\begin{flushleft}
\par\end{flushleft}

\begin{flushleft}
\par\end{flushleft}
\end{abstract}
\newpage{}

\section{Introduction}

Much of society is organized in networks and networks are important
because individuals typically interact largely or perhaps entirely
with those to whom they are closest in the network (not necessarily
physically). Examples include social networks (Facebook), professional
networks (LinkedIn), trading networks (Tesfatsion\cite{Tesfatsion}),
channels of information sharing (Chamley and Gale\cite{CG}), buyer-seller
networks (Kranton and Minehart\cite{KM}), etc. The questions which
this paper addresses are: What networks form and emerge at equilibrium
if agents are foresighted? Are these networks efficient? Are the emerging
networks robust to various types of deviations?

Network formation has been widely studied both in theoretical and
empirical settings. Starting from Jackson and Wolinsky\cite{JW} and
Bala and Goyal\cite{BG}, a large and growing theoretical literature
in economics has studied what networks form when self-interested and
strategic agents make decisions about which links to establish or
sever with other agents. Past research has mainly focused on characterizing
the networks emerging at equilibrium and determining whether socially-efficient
networks can be supported in equilibrium.

In the existing literature, network formation has been either modeled
in a static setting (e.g. Jackson and Wolinsky\cite{JW}), where agents
take actions only once and simultaneously, or a dynamic setting (e.g.
Watts\cite{Watts}) where agents meet other agents randomly over time
and choose their actions (on whether to form or break a link with
another agent) whenever they are allowed to do so. Both these model
have two important limitations: i) agents are \textit{myopic}, i.e.
their actions at any point in time are solely guided by their current
payoffs, without consideration of future consequences; and ii) agents
are \textit{homogeneous,} i.e. an agent's payoff depends only on the
network topology and her position in the network, but not on her own
characteristics or the characteristics of her peers. Such limitations
make this literature unable to model and characterize real-world social
and economic interactions. One of the key conclusions of this branch
of literature is that efficiency cannot be attained in equilibrium
and the set of achievable network topologies and that of efficient
network topologies often differ. The existing limited work on network
formation with foresighted agents (see for example Dutta et al.\cite{Dutta})
shares this negative result: they show that there are various valuation
structures in which no equilibrium can sustain efficient network topologies.
These findings are in stark contrast to empirical findings which suggest
that efficient network topologies often emerge in practice.

In this paper, we provide a first model and comprehensive analysis
of dynamic network formation that does not suffer from the abovementioned
limitations. We adopt a standard dynamic network formation game, where
agents meet one another randomly at discrete times, and they choose
whether to form or break links with other agents. As it is standard
in the network formation literature, we assume that link formation
requires bilateral consent while link severance is unilateral. We
relax two key limitations of the existing work. In our model, agents
are \textit{heterogeneous,} i.e. their payoffs depend on their own
characteristics and the characteristics of their peers as well as
the topology of the network and their position in the network. More
importantly, agents are \textit{foresighted}: in making decisions
they take into account both the current consequences of their decisions
and the future consequences of these decisions. Foresight plays an
important role in virtually all environments and would seem to play
an especially important role in network formation: agents may incur
substantial costs to form or maintain links that yield small current
benefits because they (correctly) foresee that forming these links
will encourage linking behavior of others in a way that yields large
future benefits. Analysis that ignores the effects of foresightedness
(and treats behavior as myopic) misses an important piece of agents'
cost-benefit analysis. With hindsight, it is perhaps not surprising
that such analysis -- while simpler than the analysis carried out
here -- leads to predictions that are less consistent with empirical
observations of real-life networks.

The dynamic game we analyze is a \textit{stochastic game} in which
states are determined by the random selection process as well as the
agents' actions. When types are private knowledge, this game is also
a \textit{Bayesian game,} in which agents learn dynamically and update
their beliefs based on observation of the formation history. As is
common in the literature on dynamic games, our focus is characterizing
the equilibrium behavior and outcomes when agents are patient. However,
instead of characterizing the set of achievable \textit{payoffs} as
in the repeated-game literature, we aim to characterize the set of
\textit{networks} which persist (do not change) in equillibrium (forever).

Our main findings are presented in three theorems. Theorem 1 is a
Network Convergence Theorem for the setting where agents have complete
information about characteristics of others. It demonstrates that
networks that yield each agent a positive one-period payoff can persist
in the long-run. The equilibrium strategies we construct are Markov,
and robust in several dimensions: to initial configurations, to agent
trembles, and to group deviations.

For a widely studied special case of our model, the connections model,
our results are directly comparable to previous theoretical and empirical
work, especially with respect to the sustainability of efficient networks.
Theorem 2 extends the setting of Jackson and Wolinsky\cite{JW} to
heterogeneous agents who act foresightedly. This extension yields
predictions that are much more consistent with empirical findings
than previous work. In particular, we predict that efficient networks
may be obtained in equilibrium, in contrast to \cite{JW} which predicts
that efficient networks almost never obtain. (See below for further
discussion.)

Theorem 3 is a Network Convergence Theorem for the setting in which
information is incomplete: agents begin with prior beliefs about the
characteristics of others and perform Bayesian updates based on observed
history. Under natural assumptions about the valuation structure,
there exists an equilibrium in which patient agents are incentivized
to reveal their types by making connections. In such an equilibrium,
information becomes ultimately complete and again, any network yielding
a positive one-period payoff for every agent can be sustained. This
result points to a tractable equilibrium strategy profile that covers
the range of sustainable network topologies and involves a simple
updating process.

In summary, our results yield a new and positive basis for the sustainability
of efficiency in networks: in settings in which efficient networks
provide every agent with a positive payoff, these networks \textit{can}
be sustained in equilibrium as long as agents are patients. This is
true for both the complete and incomplete information case. Nevertheless,
we find that in typical cases the agents need to be more patient under
incomplete information than under complete information to achieve
efficiency.

The remainder of the paper is organized as follows: Section 2 provides
a review of the related literature. Section 3 introduces the model.
Section 4 presents the Network Convergence Theorem under complete
information on types, with an explicit construction of equilibrium
strategies and a discussion on robustness. Section 5 analyzes the
connections model and characterizes the generically unique efficient
network topology. Section 6 introduces the Network Convergence Theorem
under incomplete information and illustrates the contrast between
complete and incomplete information. Section 7 concludes the paper.

\section{Literature Review}

The existing economics literature on network formation can be generally
categorized in two main classes: settings in which the agents link
formation is bilateral or unilateral, i.e. whether the creation of
a link requires bilateral consent of both agents involved, or can
be done unilaterally by an agent. Numerous social networks applications
such as Facebook, Google+ etc. are best modeled using models appertaining
to the first category, while Twitter is best modeled in the second
category. The well-known \textit{connections model} by Jackson and
Wolinsky\cite{JW} falls into the first category as do the models
considered in this paper. Considering static strategic environment
with homogeneous and myopic agents, \cite{JW} argued that strongly
efficient networks are necessarily either empty or a star or a clique
and that in generic cases strongly efficient networks cannot be sustained
via agents' self-interested behavior. In the second category, Bala
and Goyal\cite{BG} provided a comprehensive analysis which yields
quite different predictions on the efficient network topologies and
the equilibrium network topologies. However, we will not elaborate
on their results in this paper since our models falls into the first
category.

A more recent development in static network formation is to introduce
heterogeneity among agents. Heterogeneity takes different forms in
different branches in the literature, but it can be divided into two
main categories. The first type of heterogeneity is \textit{exogenous},
such as different failure probabilities for different links (Haller
and Sarangi\cite{HS}) and agent-specific values and costs (Galeotti\cite{Galeotti1},
Galeotti et al.\cite{Galeotti2}). The second type is \textit{endogenous}
heterogeneity, often represented as the amount of valuable resource
produced by the agents themselves, as in Galeotti and Goyal\cite{GG}.
In our paper, we adopt the first approach since the heterogeneity
we focus on is an agent's endowed individual characteristic. We assume
a more general framework than most existing literature by assigning
each agent a \textit{type} that may affect others' payoffs as well
as her own. We conduct our analysis in both cases where types are
common knowledge (complete information) and where they are private
knowledge (incomplete information).

Another strand of literature describes network formation as an interactive
process over time, instead of a one-shot, static action profile. Again,
various methods have been proposed. For instance, Johnson and Gilles\cite{JG}
and Deroian\cite{Deroian} analyze variations of the connections model
in a finite sequential game, and Konig et al.\cite{Konig} models
network formation as a continuous-time Markov chain with random arrival
of link creation opportunities. In this paper, we follow the network
formation game introduced by Watts\cite{Watts}, in which pairs of
agents are selected randomly on a discrete and infinite time line
to update the potential link between them. A link is formed or maintained
with bilateral consent, and not formed or severed if either agent
chooses to do so. This framework and variations of it have been widely
adopted to analyze strategic interactions in social and economic networks
(Jackson and Watts\cite{JWa}, Skyrms and Pemantle\cite{SP}, Song
and van der Schaar\cite{SV}).

In works adopting this dynamic model, agent myopia is a common assumption,
which means that agents only take into account their current payoffs
at every point of decision; another prevalent assumption is agent
homogeneity, with the exception of our prior work in \cite{SV}, where
we analyzed a variation of the connections model where an agent's
payoff is affected by others' types but not his own. In terms of sustaining
efficient networks, predictions made are similar to Jackson and Wolinsky
\cite{JW}: the strongly efficient network cannot be sustained at
all times in the formation process. The formation and persistence
of the strongly efficient network is random -- it depends on the realized
selection of agent pairs in the early state -- and as a result the
probability of sustaining the efficient network decreases as the number
of agents increases.

There have been a few attempts to introduce foresightedness into the
dynamic network formation, but overall this topic remains understudied.
This paper is related to Dutta et al.\cite{Dutta}, who also adopted
the model of Watts\cite{Watts} and assumed that agents take future
payoffs into account. Their main result once again points to the impossibility
of sustaining efficient networks in equilibrium by constructing a
representative example. The major difference between our paper and
this paper is that we allow for a \textit{public signal} in the definition
of a state in a Markov strategy profile -- in this way, the agents
may have only limited knowledge about the past formation history but
will still be able to cooperate in achieving efficiency. Our positive
result on sustaining efficient networks holds for a more general valuation
structure than most existing frameworks. Alternative models on foresightedness
in network formation include Page Jr. et al.\cite{Page} and Herings
et al.\cite{Herings}, whose solution concept is a pairwise stable
network instead of equilibrium. Yet again the efficiency-related results
point to cases where the strongly efficient network cannot be sustained
even if it provides each player a positive payoff.

Our approach to the analysis and in particular our construction of
equilibrium strategies owes a great deal to the work of Dutta\cite{Dutta2}
and Forges\cite{Forges}. In particular, we share the general notion
of characterizing patient agents' behavior (though in the sense of
network topologies formed instead of payoffs attained) and our construction
of equilibrium strategy profiles benefit from the existence of a ``uniform
punishment strategy'' mentioned in Forges\cite{Forges}.

There are numerous empirical studies characterizing the properties
of real-world networks. The major properties identified by these works
are: short diameter (Albert and Barabasi\cite{BA}), high clustering
(Watts and Strogatz\cite{WS}), positive assortativity (Newman\cite{Newman1}\cite{Newman2}),
and inverse relation between clustering coefficient and degree (Goyal
et al.\cite{Goyal}). Moreover, experimental studies such as Falk
and Kosfeld\cite{FK}, Corbae and Duffy\cite{CD}, Goeree et al.\cite{Goeree}
and Rong and Houser\cite{RH} have indicated that typical equilibrium
network topologies predicted by the existing theoretical analysis,
especially the star network, only emerge in a small fraction of experimental
outcomes. Last but not least, Mele\cite{Mele} and Leung\cite{Leung}
show that networks formed in large social communities, where agents
are heterogeneous and withhold certain private information, often
exhibit patterns not predicted by existing theoretical literature.
In the subsequent analysis, we will discuss most of the above properties
and illustrate how they can be accounted for in our framework.

\section{Model}

\subsection{Network Topology}

Consider a group of agents $I=\{1,2,...,N\}$. We consider undirected
networks. Thus, a $\textit{network}$ is a collection of unordered
pairs of distinct elements of \textit{I}: ${\bf {g}}\subset\mathcal{{G}}(I)=\{ij:i,j\in I,i\neq j\}$.
$ij$ is called a $\textit{link}$ between agents $i$ and $j$. A
network ${\bf {g}}$ is \textit{empty} if ${\bf {g}}=\varnothing$.
(Agents who are not linked to anyone are \textit{singletons}; in the
empty network, all agents are singletons.) Let $G(I)=2^{\mathcal{{G}}(I)}$
denote the set of all possible networks. Given a subset of agents
$I'$, let $G_{I'}$ denote a network that is formed within $I'$.

Given a network ${\bf {g}}$ we say that agents $i$ and $j$ are
$\textit{connected}$, denoted $i\overset{{\bf {g}}}{\leftrightarrow}j$,
if there exist $j_{1},j_{2},...,j_{n}$ for some $n$ such that $ij_{1},j_{1}j_{2},...,j_{n}j\in{\bf {g}}$.
Let $d_{ij}$ denote the $\textit{distance}$, or the smallest number
of links between $i$ and $j$. If $i$ and $j$ are not connected,
define $d_{ij}:=\infty$.

Let $N({\bf {g}})=\{i|\exists j\text{ s.t. }ij\in{\bf {g}}\}$ be
the set of non-singletons, and let $N_{i}({\bf {g}})=\{j:ij\in{\bf {g}}\}$
be the set of \textit{neighbors} of \textit{i}. A $\textit{component}$
of network ${\bf {g}}$ is a maximal connected sub-network, i.e. a
set $C\subset{\bf {g}}$ such that for all $i\in N(C)$ and $j\in N(C)$,
$i\neq j$, we have $i\overset{C}{\leftrightarrow}j$, and for all
$i\in N(C)$ and $j\in N({\bf {g}})$, $ij\in{\bf {g}}$ implies that
$ij\in C$. Let $C_{i}$ denote the component that contains link $ij$
for some $j\neq i$. Unless otherwise specified, in the remaining
parts of the paper we use the word ``component'' to refer to any
$\textit{non-empty}$ component.

A network ${\bf {g}}$ is said to be \textit{empty} if ${\bf {g}}=\varnothing$,
and \textit{connected} if ${\bf {g}}$ has only one component which
is itself. ${\bf {g}}$ is \textit{minimal} if for every component
$C\subset{\bf {g}}$ and every link $ij\in C$, the absence of $ij$
would disconnect at least one pair of formerly connected agents. ${\bf {g}}$
is \textit{minimally connected} if it is minimal and connected.

\subsection{Dynamic Network Formation Game}

We adopt the framework by Watts\cite{Watts} to formulate the network
formation game. Time is discrete and the horizon is infinite: $t=1,2,...$.
We assume an initial network ${\bf {g}}(0)$; this is a parameter.
The game is played as follows: 
\begin{itemize}
\item {1.} In each period, a pair of agents $(i,j)$ is randomly selected
with equal probabilities to update the link between them. 
\item {2.} The two selected agents (each knowing the identity of the other)
then play a simultaneous move game: if there is a link between them,
each can choose to sever the link or not; if there is no link between
them, each can choose to form a link or not. An existing link can
be severed unilaterally, whereas formation of a link requires mutual
consent. 
\item {3.} In addition, in each period every agent (whether or not she
is selected in the current period) can choose to sever any of her
existing links. 
\end{itemize}
It is convenient not to distinguish between severing a link and not
forming a link. Hence, for each agent \textit{i} and each agent \textit{j},
\textit{i} has two possible actions with respect to \textit{j}: $a_{ij}=1$
denotes the action that $i$ agrees to form a link with $j$ (if there
is no existing link) or not to sever the link (if there is an existing
one), and $a_{ij}=0$ otherwise. We emphasize that a link is formed
or maintained after bilateral consent (i.e. $a_{ij}=a_{ji}=1$). Write
$\mathcal{A}=\{0,1\}$.

Let $\phi(t)$ be the pair of agents selected in period $t$ and let
$\sigma(t):=\{\phi(\tau),{\bf {g}}(\tau)\}_{\tau=1}^{t}$ denote a
\textit{formation history} or a \textit{formation path} up to time
$t$, with the initial condition that $\sigma(0)=(\varnothing,\varnothing)$.
Let $\Sigma=\{\sigma(t):t\in\mathbb{N}\}$ denote the set of all possible
formation histories. It is important to note that the formation history
is different from the sequence of actions taken in two aspects. Conceptually,
the formation history is a record of the evolution of the network
from an outsider's point of view. In other words, it is the set of
all possible \textit{public information}, whereas the actions are
part of the agents' \textit{private information}. Technically, even
though the formation history is determined by actions taken over time,
it does not perfectly reveal every action. For instance, seeing a
link broken or not formed in the formation history only implies that
at least one of the two related agents chose action 0 in that period,
but it does not identify the agent(s) who did so.

Agents may not observe the entire formation history, but in each period
every agent $i$ knows its neighbors $N_{i}({\bf {g}})$, i.e. the
set of agents she links to. In addition, in each period the agents
observe a \textit{public signal} which is generated by a signal device
$y:\Sigma\rightarrow Y$, where $Y$ is the set of signal realizations.
We sometimes refer to $y$ as the \textit{monitoring structure} in
the remainder of this paper. In general, the signal generated may
depend on the entire formation history and not only on the current
actions. (Of course, the latter is a special case.) We assume that
$y$ and $Y$ are common knowledge.

The signal device determines what agents know about the formation
history. For instance, if $Y=\{0\}$ and $y(\sigma(t))=0$ then agents
have no knowledge of the formation history whatsoever; if $Y=\Sigma$
and $y(\sigma(t))=\sigma(t)$ then agents have complete knowledge
of the formation history. If agents observe the events of each period,
i.e. $(\phi(\tau),{\bf {g}}(\tau))$ for each $\tau$, then they implicitly
observe the full history $\sigma(t)$. In general however, some incomplete
monitoring structure cannot be generated by single period reports;
see our discussion in Section 5. Intermediate signals structures represent
incomplete observability. Note however, that $y$ is deterministic
instead of random, so our notion of incomplete observability is different
from the perhaps more familar notion of imperfect monitoring.

In the various applications of this model, especially social networks,
the signal device can be interpreted as a news media, e.g. a newspaper,
a television program or a website. It will not record everything in
the past for its audience, but it broadcasts important events that
attract public attention or irregular or inappropriate activities
by certain individuals. As we will see even incomplete knowledge about
the formation history is sufficient to sustain efficient networks
in equilibrium.

\subsection{Payoff Structure}

Each agent has a type, denoted by $\theta_{i}$ for agent $i$. Let
$\Theta\subset\mathbb{R}$ denote the set of possible types. Let $\bar{\theta}=(\theta_{1},\theta_{2},\cdots)$
denote the type vector for the whole group of agents. Given a subset
of agents $I'$, let $\bar{\theta}_{I'}$ denote the associated type
vector.

The one-period payoff of agent $i$ depends on the network structure
and the type vector. Specifically, this payoff is a function $u_{i}:\Theta^{N}\times G(I)\rightarrow\mathbb{R}$.
We assume that the payoff to agent $i$ is zero whenever agent $i$
is a singleton: $u_{i}(\bar{\theta},{\bf {g}})=0$ for all ${\bf {g}}$
in which $i$ is a singleton, regardless of $\bar{\theta}$ and ${\bf {g}}$.
Also, we assume that each agent's payoff satisfies \textit{component
independence}: $u_{i}(\bar{\theta},{\bf {g}})=u_{i}(\bar{\theta},C_{i})$.

For each agent, her payoff is realized in \textit{every period}, though
payoffs in different periods may well be different according to the
network topology. A payoff that realizes $t$ periods from now is
discounted by $\gamma^{t}$, where $\gamma\in(0,1)$ is the \textit{time
discount factor}. Hence, if the vector of networks that form over
time is $\bar{g}=\{{\bf {g}}(\tau)\}_{\tau=1}^{\infty}$, agent $i$'s
total (discounted) payoff evaluated at period $t$ is 
\begin{align*}
U_{i}(\bar{\theta},\bar{{\bf {g}}},t)=\sum_{\tau=0}^{\infty}\gamma^{\tau}u_{i}(\bar{\theta},{\bf {g}}(t+\tau)).
\end{align*}

If the network is constant from time $t$ onward, this reduces to
$(1-\gamma)^{-1}u_{i}(\bar{\theta},{\bf {g}}(t))$.

In our analysis, we will discuss the possibility of converging to
an efficient network structure. Following the convention in the literature,
our benchmark for efficiency will be the \textit{strongly efficient
network}, i.e. the network that yields the largest sum of one-period
payoffs. We provide a formal definition below.

\begin{defn}[Strong efficiency]Given $\bar{\theta}$, a network ${\bf {g}}$
is \textbf{strongly efficient} if $\sum_{i=1}^{N}u_{i}(\bar{\theta},{\bf {g}})\geq\sum_{i=1}^{N}u_{i}(\bar{\theta},{\bf {g}}')$
for every ${\bf {g}}'$. \end{defn}

Since the number of possible network topologies is finite, a strongly
efficient network always exists.

Another type of network we identify is a \textit{core-stable network}.
In a later section we will demonstrate that such a network entails
important additional properties of the formation process.

\begin{defn}[Core-Stable network]A network ${\bf {g}}$ is \textbf{core-stable}
if there exists no subgroup of agents $I'\subset I$ and network ${\bf {g}}'$
among $I'$ (that is, there is no link $ij\in{\bf {g}'}$ with $i\in I'$
and $j\in I\setminus I'$) such that 
\begin{align*}
u_{i}(\bar{\theta},{\bf {g}}')\geq u_{i}(\bar{\theta},{\bf {g}}),
\end{align*}
for every agent $i\in I'$, and the inequality is strict for some
agent $i\in I'$. If ${\bf {g}}$ is not core-stable, we say that
${\bf {g}}'$ \textbf{blocks} ${\bf {g}}$ and call $I'$ a \textbf{blocking
group}. \end{defn}

We use the term core-stability because this criterion discourages
any subgroup of agents to break away from the network and form a sub-network
on their own. Note that this is different from pairwise stability
defined in Jackson and Wolinsky\cite{JW}. Pairwise stability of a
network means that between any two agents, forming a new link cannot
benefit both and severing an existing link must hurt at least one,
\textit{holding other links in the network constant}. It is a widely
used solution concept for the static analysis of network formation.
Core-stability is more suitable in our dynamic setting because foresighted
agents will look ahead to the possibility of cooperation with a group
of other agents, not just a single other agent.

\subsection{Example: Connections Model}

We use the connections model in Jackson and Wolinsky\cite{JW} to
illustrate how the network topology itself, the agents' positions
and the type vector affect an agent's payoff. The connections model
is widely applied in the network formation literature.

The payoff structure in this model is described as follows. There
is a mapping $f$ from $\Theta$ to $\mathbb{R}^{++}$ that specifies
payoffs from direct connections: if agent $i$ is directly connected
with agent $j$ ($ij\in{\bf {g}}$), then agent $i$ gets payoff $f(\theta_{j})$
and agent $j$ gets payoff $f(\theta_{i})$ from this connection.
In addition, if agent $i$ is indirectly connected to agent $j$,
then $i$ obtains the payoff $f(\theta_{j})$ discounted by $\delta^{d_{ij}-1}$,
where $\delta\in(0,1)$ is the \textit{spatial discount factor}, and
$d_{ij}$ is the \textit{distance} between $i$ and $j$ measured
in the number of links. Finally, agent $i$ pays a cost of $c>0$
per period for every link that $i$ has. Hence, in a single period
with network ${\bf {g}}$, agent $i$'s current payoff is 
\begin{align*}
u_{i}(\bar{\theta},{\bf {g}})=\sum_{j:i\overset{{\bf {g}}}{\leftrightarrow}j}\delta^{d_{ij}-1}f(\theta_{j})-\sum_{j:ij\in{\bf {g}}}c.
\end{align*}

It is easy to see that the above payoff structure satisfies the assumptions
we made in the previous section. In the original and widely adopted
version of the model, agents are assumed to be homogeneous, i.e. $f(\theta_{j})$
is a constant independent of $j$. In the above formulation, the agents
are heterogeneous: an agent's payoff obtained from a connection depends
on the type of the agent it connects to. Note that the payoff structure
exhibits non-local externalities: though an agent gets a positive
payoff from each agent she connects to, she only pays a cost for each
link she maintains. Moreover, an agent's payoff depends both on the
network topology as well as an agent's position. In particular, agents
who are distantly connected obtain lower payoffs from their connection
than agents who are closely connected. In various applications, this
spatial discount can be regarded as the decay of a valuable resource
or information due to increased noise or risk. In a later section,
we will discuss the connections model in more details and present
important related results.

\section{Network Convergence Theorem with Complete Information}

In this section, we characterize the set of networks that can persist
in equilibrium when agents are patient, assuming that the type vector
is commonly known. We start by defining strategies in this environment
and the concept of an equilibrium. In particular, we are interested
in equilibria in which the network formation process converges, i.e.
over time the network rests on a specific topology which then persists
forever.

\subsection{Strategy, Equilibrium and Convergence}

Fix the signal structure $Y$. A (pure) \textit{strategy} of agent
$i$ is a mapping that assigns, following every history, an action
in $\{0,1\}$ to every other agent $j$. The constraint on this mapping
is that if $i$ and $j$ are not linked and the pair $(i,j)$ is not
selected in the current period, then agent $i$'s action towards agent
$j$ has to be $0$. Formally, let $\omega_{ij}\in\Omega=\{0,1\}$
denote the state of whether the pair $(i,j)$ is selected, and let
$\zeta_{ij}\in Z=\{0,1\}$ denote the state of whether $i$ and $j$
are linked in the current period. Write $\mathcal{H}$ for the set
of histories of public signals. \begin{defn}[Strategy]A (pure) \textbf{public
strategy} of agent $i$ is a mapping $s_{i}$: 
\begin{align*}
s_{i}:(I-i)\times\mathcal{{H}}\times\Omega\times Z\rightarrow\mathcal{A}
\end{align*}
such that $s_{ij}^{y}(\cdot,\cdot,0,0)\equiv0$. \end{defn} Let $S$
denote the set of all public strategies. (As is customary, we assume
that agents condition only on the public signal.)

Throughout the paper, we will focus on \textit{Markov strategies},
which by definition depend not on the entire history of signals but
only on the current signal. Hence, a Markov strategy is a mapping
$s_{i}:(I-i)\times Y\times\Omega\times Z\rightarrow\mathcal{A}$.

Associated with the device for public signals, the interpretation
of a strategy in this game is rather straightforward. For every agent
$i$, the \textit{state} in a Markov strategy at a given time period
is represented by her knowledge about the game at that period, which
is the combination of two elements: her knowledge about every other
agent $j$, which includes $j$'s type and whether $j$ is linked
to herself; and her knowledge about the formation history $\sigma(t)$.
$i$'s information on the former is complete since she knows both
the identity of $j$ and $j$'s type. The precision of her information
on the latter, on the other hand, may vary according to the public
signal generating function $y$. Note that strategies thus defined
include strategies that assign actions only based on the network formed
in the previous period (so that $Y=G$) as in some existing literature,
for instance Dutta et al.\cite{Dutta}. A profile of strategies and
a history define a probability distribution on future histories assuming
agents follow the given strategies. (Randomness arises because the
selection process is random.) When we take expectations we implicitly
mean expectations with respect to this probability distribution.

Now we are ready to define the equilibrium.

\begin{defn}[Equilibrium]A (pure strategy) \textbf{public perfect
Markov equilibrium} is a vector of public Markov strategies $s^{*}=(s_{1}^{*},\cdots,s_{N}^{*})$
such that: for each agent $i$, every period $t$ and every possible
history of the public signals, $s_{i}^{*}$ maximizes agent $i$'s
expected discounted total payoff at period $t$ given $s_{-i}^{*}$.
\end{defn}

For the remainder of the paper, we simply refer to a public perfect
Markov equilibrium as an equilibrium. It is easy to see that a pure
strategy equilibrium for the game always exists, regardless of the
type vector and the specific payoff structure. Indeed, since link
formation and maintenance requires bilateral consent, the strategy
profile that every agent always chooses action $0$ (sever/not form
a link) already constitutes an equilibrium. We note the existence
of an equilibrium below.

\begin{prop} There exists a pure strategy equilibrium. \end{prop}

We focus on equilibria in which the network formation process converges
(after a finite number of periods) and so leads to a persisting network.
Before convergence occurs, the evolution of the network is random
because the selection process is random; after convergence occurs,
randomness has no further effect and so the limit network is a random
function of the initial network and the strategies. We believe that
this notion provides an appropriate account for what is to be expected
in the formation process in various applications such as social circles.
People tend to form and sever links constantly in the starting phase
of building their social milieu, but over time they maintain a relatively
fixed circle of acquaintances (Kossinets and Watts\cite{KW}). We
formally describe such convergence in our model below.

Given a realized formation history $\sigma(t)$ , the network topology
thereafter $\{{\bf {g}}(\tau)\}_{\tau=t+1}^{\infty}$ is a stochastic
process. We denote the probability measure generated by this stochastic
process as $\mathcal{Q}_{s^{*},\sigma(t)}$.

\begin{defn}[Convergence]Given a realized formation history $\sigma(t)$
we say that the network formation process \textbf{converges weakly}
to network ${\bf {g}}$ in equilibrium $s^{*}$ if 
\begin{align*}
\lim_{T\rightarrow\infty}\mathcal{Q}_{s^{*},\sigma(t)}({\bf {g}}(T')={\bf {g}}\text{ }\forall T'\geq T)=1.
\end{align*}

We say that the network formation process \textbf{converges strongly}
to network ${\bf {g}}$ if it converges following \textbf{\textit{every}}
(finite) history. \end{defn}

Notice that convergence entails that the network converges in finite
time with probability 1.

In what follows we focus on strong convergence rather than weak convergence
for 2 reasons. The first is that strong convergence implies that if
the evolution of the network is disturbed by some exogenous process
then it eventually returns to the same limit. The second is that strong
convergence guarantees robustness with respect to small errors and
with respect to coalitional deviations, not just individual deviations.
We will discuss these points in more detail below.

\subsection{Informative Monitoring Structures}

We will explicitly construct equilibrium strategy profiles that yield
strong convergence to a given network provided that the monitoring
is ``sufficiently informative''. We begin by describing what this
entails.

Fix a network ${\bf {g}}$ and integer $K\geq1$. We begin by defining
a particular monitoring structure $y_{{\bf {g}},K}$. 
\begin{itemize}
\item {1.} $Y=\{C,P\}$, where $C$ represents the \textit{cooperation
phase} and $P$ represents the \textit{punishment phase}. 
\item {2.} $y{}_{{\bf {g}},K}(\sigma(0))=C$. 
\item {3.} In period $t\geq1$: if $y{}_{{\bf {g}},K}(\sigma(t-1))=C,$
we distinguish 2 cases:

\begin{itemize}
\item case 1: for every pair of agents $ij\in{\bf {g}}$, $a_{ij}=a_{ji}=1$
and for every pair of agents $ij\notin{\bf {g}}$ , $a_{ij}=0$ or
$a_{ji}=0$ (or both). 
\item case 2: otherwise (i.e. case 1 fails for some pair of agents $ij$) 
\item In case 1, we define $y_{{\bf {g}},K}(\sigma(t))=C$ and in case 2,
we define $y{}_{{\bf {g}},K}(\sigma(t))=P$. 
\end{itemize}
\item {4.} In period $t\geq1$, if $y{}_{{\bf {g}},K}(\sigma(t-1))=P$:
we again distinguish 2 cases:

\begin{itemize}
\item case 1: $y(\sigma(t-2))=y(\sigma(t-3))=.....y(\sigma(t-K))=P$ 
\item case 2: otherwise 
\item In case 1, we define $y{}_{{\bf {g}},K}(\sigma(t))=C$ and in case
2, $y{}_{{\bf {g}},K}(\sigma(t))=P$. 
\end{itemize}
\end{itemize}
As we will see in the proof of Theorem 1 below access to the information
provided by $\ensuremath{y_{{\bf {g}},K}}$ allows the agents to divide
the formation process into two phases: the cooperation phase which
continues forever if agents choose their actions in order to form
or maintain the network ${\bf {g}}$, and the punishment phase that
starts when agents depart from the cooperation phase and continues
for $K$ periods. From the public signal $C$ or $P$, each agent
knows what phase she should currently be in, but not how long that
phase has lasted or how many times the same phase has occurred before.
As we will also see in the proof, the parameter $K$ plays a crucial
role in guaranteeing that convergence is strong rather than weak.

Consider any other signal structure $\hat{y}$ with signal space $\hat{Y}$.
$\hat{y}$ is as informative as $y_{{\bf {g}},K}$ if there is a mapping
$\eta:\hat{Y}\to Y$ such that $y_{{\bf {g}},K}(\sigma(t))=\eta(\hat{y}(\sigma(t))$.
That is, $\hat{y}$ reveals at least as much about the history as
$y_{{\bf {g}},K}$ (and perhaps more). Notice that complete information
is always as informative as $y_{{\bf {g}},K}$, no matter what ${\bf {g}\ensuremath{}}$
and $K$ are.

\subsection{Construction of Equilibrium Strategies}

It is useful to give an explicit description of the strategies we
will use in the proof. We assume that the monitoring structure y is
as informative as $y_{{\bf {g}},K}$. Hence, agents always know what
they would know if the monitoring structure were exactly $y_{{\bf {g}},K}$;
the strategies we describe use only this information, so there is
no loss in assuming that the monitoring structure is exactly $y_{{\bf {g}},K}$.

Consider the following strategy profile, denoted $\hat{s}_{c}$:

\begin{center}
$s_{ij}=\left\{ \begin{array}{c}
1,\text{ if \ensuremath{ij\in{\bf {g}}}, \ensuremath{y_{{\bf {g}},K}=C}, and \ensuremath{\max\{\omega_{ij},\zeta_{ij}\}=1};}\\
0,\text{ otherwise.}
\end{array}\right.$ 
\par\end{center}

$\hat{s}_{c}$ can be interpreted as the following pattern of behavior:
the agents start by cooperating towards building a designated network.
They form or maintain a link if and only if that link belongs to the
specific network ${\bf {g}}$. If a ``deviation'' - a link in ${\bf {g}}$
is not formed or a link not in ${\bf {g}}$ is formed - is detected
all agents leave the social circle (break all links) for $K$ periods
before starting cooperation again.

\subsection{The Network Convergence Theorem}

We begin with a simple observation.

\begin{prop}If there exists an equilibrium in which the formation
process converges weakly to the network ${\bf g}$, then $u_{i}(\bar{\theta},{\bf {g}})\geq0$
.\end{prop}

\begin{proof}Suppose that there exists an equilibrium where the formation
process converges to ${\bf {g}}$ weakly, and that $u_{i}(\bar{\theta},{\bf {g}})<0$
for some $i$. Then on the equilibrium path when ${\bf {g}}$ has
been formed and will persist forever, $i$ is always strictly better
off by deviating to the strategy $s_{ij}=0$ and obtaining payoff
$0$ thereafter. This is a contradiction to the assumption of an equilibrium.
\end{proof}

The Network Convergence Theorem with complete information shows that
if the inequality is strict for all agents i, the monitoring structure
is sufficiently informative (in particular if the monitoring structure
yields complete information) and agents are sufficiently patient then
there is an equilibrium in which the formation process converges strongly
to ${\bf {g}}$.\begin{thm}Let ${\bf {g}}$ be a network for which
$u_{i}(\bar{\theta},{\bf {g}})>0$ for all $i$. There is an integer
$K$ and a cutoff $\bar{\gamma}\in(0,1)$ such that if $\gamma\in[\bar{\gamma},1)$
and the monitoring structure is as informative as $y_{{\bf {g}},K}$,
then there exists an equilibrium in which the formation process converges
strongly to ${\bf {g}}$.\end{thm}

As we have noted above, complete monitoring is always as informative
as $y_{{\bf {g}},K}$ so we obtain as an immediate corollary the corresponding
Network Convergence Theorem for complete monitoring.

\begin{cor} Let ${\bf {g}}$ be a network for which $u_{i}(\bar{\theta},{\bf {g}})>0$
for all $i$. If monitoring is complete, there is a cutoff $\bar{\gamma}\in(0,1)$
such that if $\gamma\in[\bar{\gamma},1)$, then there exists an equilibrium
in which the formation process converges strongly to ${\bf {g}}$.\end{cor}

It is useful to contrast the Network Convergence Theorem with the
familiar Folk Theorem for the repeated games. The Folk Theorem says
that every feasible, strictly individually rational long-run average
payoff vector can be achieved in an equilibrium if agents are sufficiently
patient. The Network Convergence Theorem says that every ``feasible,
strictly individually rational'' network can be achieved as the limit
of a formation process. The Folk Theorem talks about the long-run
payoffs; the Network Convergence Theorem talks about the long-run
network. The proof of the theorem is stated below.

\begin{proof}Consider the monitoring structure $y_{{\bf {g}},K}$
and the strategy profile $\hat{s}_{c}$. Given $\bar{\theta}$, let
$\bar{v}$ denote the largest marginal benefit that an agent can obtain
from forming or severing a link in any network ${\bf {g}}$. $\bar{v}$
measures the largest possible marginal benefit that an agent can get
from deviating in one period. Since the number of networks is finite,
we know that $\bar{v}$ exists.

Given $\bar{\theta}$and a formation history $\sigma(t)$, consider
an arbitrary agent. Let $\bar{\mu}_{C}(\gamma,M)$ and $\underline{\mu}_{C}(\gamma,M)$
denote the largest and smallest expected total payoff the agent gets
within $M$ periods of the cooperation phase, starting from any network.
Note that an agent's payoff during the punishment phase is always
equal to $0$. We first establish the following lemma.

\begin{lem} If $u_{i}(\bar{\theta},{\bf {g}})>0$ for all $i$, then
the following properties hold: 
\begin{itemize}
\item {a.} $\lim_{\gamma\rightarrow1}\bar{\mu}_{C}(\gamma,\infty)=\lim_{\gamma\rightarrow1}\underline{\mu}_{C}(\gamma,\infty)=\infty$. 
\item {b.} There exists $A>0$ such that $\bar{\mu}_{C}(\gamma,\infty)-\underline{\mu}_{C}(\gamma,\infty)<A$,
regardless of $\gamma$. 
\end{itemize}
\end{lem}

\begin{proof} Let $W\in\mathbb{R}$ denote the smallest possible
payoff of any agent in any network in one period.

For $(a)$, it suffices to show that a lower bound of the two payoffs
converges to infinity as $\gamma$ converges to $1$. Consider the
following hypothetical payoff structure: agent $i$'s one-period payoff
is $W$ if the network is different from ${\bf {g}}$, and $u_{i}(\bar{\theta},{\bf {g}})$
otherwise. Starting from any network ${\bf {g}}(0)$, the probability
that ${\bf {g}}(t)\neq\{{\bf {g}}\}$ is bounded above by $\min\{\frac{N(N-1)}{2}(1-\frac{2}{N(N-1)})^{t},1\}$
(this upper bound is constructed by supposing that ${\bf {g}}(0)=\varnothing$
and ${\bf {g}}$ is the complete network, and calculating the probability
that some pair of agents has never been selected during the $t$ periods).
For all $t$ such that $\frac{N(N-1)}{2}(1-\frac{2}{N(N-1)})^{t}<1$
(let $t^{*}$ be the smallest $t$ satisfying this condition), $i$'s
expected payoff in ${\bf {g}}(t)$ is bounded below by 
\begin{align*}
\frac{N(N-1)}{2}(1-\frac{2}{N(N-1)})^{t}W+(1-\frac{N(N-1)}{2}(1-\frac{2}{N(N-1)})^{t})u_{i}(\bar{\theta},{\bf {g}}),
\end{align*}
and agent $i$'s total expected payoff is bounded below by 
\begin{align*}
 & \sum_{t=1}^{t^{*}-1}\gamma^{t-1}W+\sum_{t=t^{*}}^{\infty}\gamma^{t-1}(\frac{N(N-1)}{2}(1-\frac{2}{N(N-1)})^{t}W\\
 & +(1-\frac{N(N-1)}{2}(1-\frac{2}{N(N-1)})^{t})u_{i}(\bar{\theta},{\bf {g}}))\\
= & \sum_{t=1}^{t^{*}-1}\gamma^{t-1}W+\sum_{t=t^{*}}^{\infty}\gamma^{t-1}u_{i}(\bar{\theta},{\bf {g}})+\sum_{t=t^{*}}^{\infty}\gamma^{t-1}\frac{N(N-1)}{2}(1-\frac{2}{N(N-1)})^{t}(W-u_{i}(\bar{\theta},{\bf {g}}))\\
= & \frac{W(1-\gamma^{t^{*}})}{1-\gamma}+\frac{\gamma^{t^{*}-1}u_{i}(\bar{\theta},{\bf {g}})}{1-\gamma}+\frac{N(N-1)}{2}\frac{\gamma^{t^{*}-1}(1-\frac{2}{N(N-1)})^{t^{*}}(W-u_{i}(\bar{\theta},{\bf {g}}))}{1-\gamma(1-\frac{2}{N(N-1)})}.
\end{align*}

It is clear that the sum of the first term and the third term above
has a lower bound which is independent of $\gamma$. In addition,
the second term converges to infinity as $\gamma$ converges to $1$
regardless of $i$. Hence part $(a)$ is proved. $(b)$ can be proved
using a similar argument. \end{proof}

Consider agent $i$ at period $t$ following any formation history.
Note that $i$ cannot really ``deviate'' in the punishment phase
given that all the agents other than $i$ are using their prescribed
strategy in $\hat{s}_{c}$. Hence we only need to consider a deviation
of agent $i$ in the cooperation phase. According to the one-step
deviation principle, in order to determine whether $\hat{s}_{c}$
is an equilibrium we only need to consider $i$'s deviation in one
period, after which $i$ returns to her prescribed strategy in $\hat{s}_{c}$.
As mentioned before, the largest possible marginal benefit that $i$
gets from this deviation in this period is $\bar{v}$. Starting from
the next period, $i$'s expected total payoff is bounded above by
\begin{align*}
\gamma^{1+K}\bar{\mu}_{C}(\gamma,\infty).
\end{align*}
If $i$ does not deviate, then starting from the next period, $i$'s
expected total payoff is bounded below by 
\begin{align*}
\gamma\underline{\mu}_{C}(\gamma,K)+\gamma^{1+K}\underline{\mu}_{C}(\gamma,\infty).
\end{align*}
Therefore, we have 
\begin{align*}
 & \text{Total expected marginal benefit from deviation}\\
\leq & \bar{v}+\gamma^{1+K}(\bar{\mu}_{C}(\gamma,\infty)-\underline{\mu}_{C}(\gamma,\infty))-\gamma\underline{\mu}_{C}(\gamma,K)\\
< & \bar{v}+\gamma^{1+K}A-\gamma\underline{\mu}_{C}(\gamma,K),
\end{align*}
from property $(b)$ above. Then from property $(a)$, there exists
$\gamma'\in(0,1)$ and $\bar{K}$ such that $\underline{\mu}_{C}(\gamma,K)>2(\bar{v}+A)$
for every $\gamma\geq\gamma'$ and $K\geq\bar{K}$. Let $\bar{\gamma}=\max\{\gamma',\frac{1}{2}\}$,
then for every $\gamma\in[\bar{\gamma},1)$, we have 
\begin{align*}
\bar{v}+\gamma^{1+K}A-\gamma\underline{\mu}_{C}(\gamma,K)<0,
\end{align*}
which implies that deviation is not profitable and hence $\hat{s}_{c}$
is an equilibrium where the formation process converges to ${\bf {g}}$.
\end{proof}

It might be noted that Proposition 2 and Theorem 1 together do not
quite provide a complete characterization of which networks can be
achieved as weak or strong limits: if $u_{i}(\bar{\theta},{\bf {g}})=0$
for some i, the results are silent about the achievability of ${\bf {g}}$.
This seems entirely analogous to the situation for the familiar Folk
Theorem: strictly individually rational payoff vectors can be achieved
and sub-rational payoff vectors cannot be achieved but the status
of payoff vectors that are exactly rational (i.e. equal to the minmax
payoff) is indeterminate.

With our proposed strategy profile, the underlying mechanism for convergence
to such a network can be described as a ``self-fulfilling prophecy'':
the agents cooperate in order to form a network that is commonly envisioned,
and they punish any detected deviation (there can be ``undetected''
deviations such as $i$ choosing $1$ but $j$ still choosing $0$
for a link $ij\notin{\bf {g}}$) by opting out of the group for $K$
periods. for every agent, this punishment is incentive compatible
once everyone else complies. Afterwards, the agents opt back in and
resume cooperation. In this way, ${\bf {g}}$ always gets formed and
persists no matter what the initial network was and what the formation
history has been.

An immediate yet important result from Theorem 1 is a clear criterion
on sustaining efficiency. For a strongly efficient network to be sustained
in any equilibrium, it needs to ensure a non-negative payoff for each
agent. Conversely, if a strongly efficient network yields every agent
a positive payoff, then it can be sustained in equilibrium if the
agents are patient enough.

\begin{cor} If ${\bf {g}}$ is a strongly efficient network, and
there is an equilibrium in which the formation process converges weakly
to ${\bf {g}}$ then $u_{i}(\bar{\theta},{\bf {g}})\geq0$ for all
$i$. If ${\bf {g}}$ is strongly efficient, $u_{i}(\bar{\theta},{\bf {g}})>0$
for all $i$, and agents are sufficiently patient, then there exists
an equilibrium in which the formation process converges strongly to
${\bf {g}}$ . \end{cor}

This corollary presents a striking contrast to the argument offered
by Dutta et al.\cite{Dutta} that in generic cases efficiency cannot
be sustained even if agents are patient and each agent's payoff in
the strongly efficient network is positive. We provide an example
below, which is taken from Dutta et al.\cite{Dutta}, to illustrate
the difference.

\begin{exmp} This example is taken from Dutta et al., Theorem 2\cite{Dutta}.
Consider $I=\{1,2,3\}$ and assume that all agents are of the same
type. The payoff structure is symmetric: for every $i,j,k$, $u_{i}(\bar{\theta},\varnothing)=0$,
$u_{i}(\bar{\theta},\{ij\})=2v$, $u_{i}(\bar{\theta},\{ij,ik,jk\})=v$,
while $u_{i}(\bar{\theta},\{ij,jk\})=0$. The unique strongly efficient
network is the complete network ${\bf {g}}=\{12,13,23\}$.

\cite{Dutta} shows that there exists $\bar{\gamma}<1$ such that
if $\gamma\in(\bar{\gamma},1)$ then there is no pure strategy equilibrium
where the formation process converges strongly to the strongly efficient
network. This results from the constraint of agents' knowledge on
the formation history: in \cite{Dutta} it is assumed that the agents
only know the network formed in the previous period and the pair of
agents selected in the current period.

In our model, agents know more and this matters. To see why, consider
the strategy profile $\hat{s}_{c}$. In the punishment phase, no unilateral
action can change the network formation outcome, so we only need to
inspect the incentives of agents to deviate in the cooperation phase.
Using the same methods as in the proof of Theorem 1 and plugging in
the values in this example, we can obtain a range of $\gamma$ and
$K$ to make $\hat{s}_{c}$ an equilibrium: $\gamma\in(0.97,1)$,
$K\geq60$. \end{exmp}

At the end of this section we would like to emphasize again the importance
and significance of the monitoring structure. To sustain cooperation
which leads to efficiency over time, it is not necessary that agents
know everything. The agents need not know \textit{who} committed a
deviation or \textit{when} a deviation occurred, but it is vital that
they know if \textit{someone} has deviated in the recent past and
\textit{whether} they are supposed to carry out punishment. A public
signal device (newspaper, TV, website, etc.) can convey such information
across the group of agents and ensure a limited but effective form
of cooperation. As a practical implication, our analysis strongly
suggests that modern media, with its function of public broadcast,
plays a crucial role in enhancing social welfare.

\subsection{Robustness of Equilibrium}

Agents are not always rational and do not always choose actions independently
of others, so it seems important to ask whether results such as ours
are robust to ``mistakes'' and to coalitional deviations (in addition
to individual deviations). In this section, we demonstrate robustness
of our results with respect to individual mistakes and coalitional
deviations.

We consider a model in which agents tremble uniformly. Fix a strategy
profile $s$ and fix $\epsilon>0$. Write $s^{\epsilon}$ for the
mixed strategy profile in which each agent $i$ plays $s_{i}$ with
probability $(1-\epsilon)$ and chooses a random action with probability
$\epsilon$. Let $Q^{\epsilon}$ be the probability distribution on
the corresponding stochastic process of networks. Intuitively, if
$\epsilon$ is sufficiently small, then the network formation process
will lead to ${\bf {g}}$ but will not remain there because agents
will randomly break links in ${\bf {g}}$ ``by accident'', However,
following such a breakage the process will lead back to ${\bf {g}}$.
Hence, if $\epsilon$ is small, ${\bf {g}}$ will probably occur ``most
of the time''. The following proposition formalizes this result.

\begin{prop}

Fix a network ${\bf {g}}$ such that $u_{i}(\bar{\theta},{\bf {g}})>0$
for all $i$ and an integer $K$ and a monitoring structure $y$ that
satisfy the conditions of Theorem 1. Fix $a,b>0$. There exists $\bar{\gamma}\in(0,1)$
such that if $\gamma\in[\bar{\gamma},1)$ and $\hat{s}_{c}$ is the
corresponding equilibrium strategy constructed in Theorem 1, then

\[
\lim_{T\rightarrow\infty}\inf Q^{\epsilon}(\frac{|\{t:1\leq t\leq T,{\bf {g}}(t)={\bf {g}}\}|}{T}>1-a)>1-b
\]
for all sufficiently small $\epsilon$.

\end{prop}

\begin{proof} We first prove that $\hat{s}_{c}$ is an equilibrium
when $\epsilon$ is sufficiently small. Referring to the proof of
Theorem 1, it suffices to show that Lemma 1 still holds under this
alternative environment with a function $\epsilon(\gamma)$. For part
$(a)$ of Lemma 1, let $t^{*}$ be the smallest $t$ satisfying $\frac{N(N-1)}{2}(1-\frac{2}{N(N-1)})^{t}<1$,
and following a similar argument as in the proof of Theorem 1, we
know that $i$'s expected payoff in ${\bf {g}}(t)$ is bounded below
by 
\begin{align*}
 & W(1-(1-\epsilon)^{Nt})+\frac{N(N-1)}{2}(1-\frac{2}{N(N-1)})^{t}W(1-\epsilon)^{Nt}\\
 & +(1-\frac{N(N-1)}{2}(1-\frac{2}{N(N-1)})^{t})u_{i}(\bar{\theta},{\bf {g}})(1-\epsilon)^{Nt}.
\end{align*}
Agent $i$'s total expected payoff is bounded below by 
\begin{align*}
 & \sum_{t=1}^{t^{*}-1}\gamma^{t-1}W+\sum_{t=t^{*}}^{\infty}\gamma^{t-1}W(1-(1-\epsilon)^{Nt})+\sum_{t=t^{*}}^{\infty}\gamma^{t-1}\frac{N(N-1)}{2}(1-\frac{2}{N(N-1)})^{t}W(1-\epsilon)^{Nt}\\
 & +\sum_{t=t^{*}}^{\infty}(1-\frac{N(N-1)}{2}(1-\frac{2}{N(N-1)})^{t})u_{i}(\bar{\theta},{\bf {g}})(1-\epsilon)^{Nt}\\
= & \frac{W}{1-\gamma}-\frac{\gamma^{t^{*}-1}(1-\epsilon)^{Nt^{*}}W}{1-\gamma(1-\epsilon)^{N}}+\frac{\gamma^{t^{*}-1}(1-\epsilon)^{Nt^{*}}u_{i}(\bar{\theta},{\bf {g}})}{1-\gamma(1-\epsilon)^{N}}\\
 & +\frac{N(N-1)}{2}\frac{\gamma^{t^{*}-1}(1-\epsilon)^{Nt^{*}}(1-\frac{2}{N(N-1)})^{t^{*}}(W-u_{i}(\bar{\theta},{\bf {g}}))}{1-\gamma(1-\frac{2}{N(N-1)})(1-\epsilon)^{N}}.
\end{align*}

The last term has a lower bound which is independent of $\gamma$.
Also, note that as $\epsilon\rightarrow0$, we have 
\begin{align*}
\frac{W}{1-\gamma}-\frac{\gamma^{t^{*}-1}(1-\epsilon)^{Nt^{*}}W}{1-\gamma(1-\epsilon)^{N}} & \rightarrow\frac{W}{1-\gamma}-\frac{\gamma^{t^{*}-1}W}{1-\gamma}\\
\frac{\gamma^{t^{*}-1}(1-\epsilon)^{Nt^{*}}u_{i}(\bar{\theta},{\bf {g}})}{1-\gamma(1-\epsilon)^{N}} & \rightarrow\frac{\gamma^{t^{*}-1}u_{i}(\bar{\theta},{\bf {g}})}{1-\gamma},
\end{align*}
and as $\gamma\rightarrow1$, we have 
\begin{align*}
\frac{W}{1-\gamma}-\frac{\gamma^{t^{*}-1}W}{1-\gamma} & \rightarrow t^{*}-1\\
\frac{\gamma^{t^{*}-1}u_{i}(\bar{\theta},{\bf {g}})}{1-\gamma} & \rightarrow\infty.
\end{align*}
Hence, for every number $x>0$, there exists a function $\epsilon(\gamma)$
such that for some $\gamma'\in(0,1)$, any $(\gamma,\epsilon)$ such
that $\gamma\geq\gamma'$ and $\epsilon\leq\epsilon(\gamma)$ makes
agent $i$'s expected total payoff higher than $x$. This proves part
$(a)$ of Lemma 1. Part $(b)$ can be proved by a similar argument.

Now we prove the limit inferior in probability. Let $\hat{T}$ be
a sufficiently large integer such that $\frac{\hat{T}-K}{\hat{T}}>1-a$.
We know that for every $b>0$ and $\tau\in\mathbb{N}^{+}$, when $\epsilon$
is sufficiently small we have $Q^{\epsilon}(\frac{|\{t:\tau\leq t\leq\tau+\hat{{T}},{\bf {g}}(t)={\bf {g}}\}|}{\hat{{T}}}>1-a)>1-b$.
Moreover, this property is invariant for every time period of length
$n\hat{T}$, $n\in\mathbb{N}^{+}$. This completes the proof.\end{proof}

Next, we discuss how the additional property of stability of a network
brings about an equilibrium that prevents typical \textit{group} deviations.
Recall that a network is core-stable if there is no subgroup of agents
that can form another network on their own (without linking to any
agent not in the subgroup) and provide a Pareto improvement for the
subgroup. We consider a natural class of group deviations. Fix a sub-group
of agents $\hat{I}\subsetneq I$ and a network $\hat{{\bf {g}}}\in G(\hat{I})$
. We consider a group deviation by agents in $\hat{I}$ in which they
commit to forming $\hat{{\bf {g}}}$ (that is they agree to form or
maintain link $ij$ if and only if $ij\in{\bf \hat{{\bf {g}}}}$).
We refer to these deviations as \textit{network deviations}.

In the following result, we assume that the monitoring structure reveals
the remaining number of periods for the punishment phase.

\begin{prop}Fix a core-stable network ${\bf {g}}$ such that $u_{i}(\bar{\theta},{\bf {g}})>0$
for all $i$. There exists an integer $\hat{K}$ and a cutoff $\hat{\gamma}\in(0,1)$
such that for every $\gamma\in[\hat{\gamma},1)$, there exists $M(\gamma)$
such that if $K\geq\hat{K}$ then 
\begin{itemize}
\item {a.} The strategy profile constructed in Theorem 1 is an equilibrium
and the formation process converges strongly to ${\bf g}$. 
\item {b.} Following any formation history with the remaining punishment
phase no longer than $M(\gamma)$ periods, no proper subgroup of agents
has a profitable network deviation. 
\item {c.} $M(\gamma)$ is increasing in $\gamma$ and $\lim_{\gamma\rightarrow1}M(\gamma)=\infty$. 
\end{itemize}
\end{prop}

\begin{proof} Consider a formation history with the remaining punishment
phase being $K'$ periods (including the current period). By the assumption
that ${\bf {g}}$ is core-stable, for every $\hat{I}\subsetneq I$
and associated $\hat{{\bf {g}}}$, there exists an agent $i$ such
that $u_{i}(\bar{\theta},\hat{{\bf {g}}})-u_{i}(\bar{\theta},{\bf {g}})<0$.
Fix one such $i$. From the current period onwards, if the agents
follow $\hat{s}_{c}$, $i$'s total payoff is bounded below by 
\begin{align*}
\gamma^{K'}\underline{\mu}_{C}(\gamma,\infty).
\end{align*}
Let $V>0$ denote the largest possible payoff of any agent in any
network in one period. With a little abuse of notation, let ${\bf {g}}(t)$
denote the network formed $t$ periods from the current period. If
the group of agents $\hat{I}$ follow $s'(\hat{I},\hat{{\bf {g}}})$,
$i$'s payoff in ${\bf {g}}(t)$ is bounded above by 
\begin{align*}
1\{{\bf {g}}(t)\neq\hat{{\bf {g}}}\}V+(1-1\{{\bf {g}}(t)\neq\hat{{\bf {g}}}\})u_{i}(\bar{\theta},\hat{{\bf {g}}}).
\end{align*}
If the group of agents $\hat{I}$ follow $s'(\hat{I},\hat{{\bf {g}}})$,
the probability that ${\bf {g}}(t)\neq\hat{{\bf {g}}}$ is bounded
above by $\min\{\frac{N(N-1)}{2}(1-\frac{2}{N(N-1)})^{t+1},1\}$.
For all $t$ such that $\frac{N(N-1)}{2}(1-\frac{2}{N(N-1)})^{t+1}<1$,
$i$'s expected payoff in ${\bf {g}}(t)$ is bounded above by 
\begin{align*}
\frac{N(N-1)}{2}(1-\frac{2}{N(N-1)})^{t+1}V+(1-\frac{N(N-1)}{2}(1-\frac{2}{N(N-1)})^{t+1})u_{i}(\bar{\theta},\hat{{\bf {g}}}).
\end{align*}
Following a similar argument to the proof of Lemma 1, there exists
$D>0$ (regardless of $\gamma$) such that $i$'s discounted expected
total payoff from $s'(\hat{I},\hat{{\bf {g}}})$ is less than $D+\sum_{t=0}^{\infty}\gamma^{t}u_{i}(\bar{\theta},\hat{{\bf {g}}})$.
Now, the difference in $i$'s payoff between the two strategy profiles
is bounded above by 
\begin{align*}
D+\sum_{t=0}^{\infty}\gamma^{t}u_{i}(\bar{\theta},\hat{{\bf {g}}})-\gamma^{K'}\underline{\mu}_{C}(\gamma,\infty).
\end{align*}
With a similar argument to above, there exists $E>0$ (regardless
of $\gamma$) such that $\underline{\mu}_{C}(\gamma,\infty)>\sum_{t=0}^{\infty}\gamma^{t}u_{i}(\bar{\theta},{\bf {g}})-E$.
Hence, the difference in $i$'s payoff between the two strategy profiles
is bounded above by 
\begin{align*}
 & D+E+\sum_{t=0}^{\infty}\gamma^{t}u_{i}(\bar{\theta},\hat{{\bf {g}}})-\gamma^{K'}\sum_{t=0}^{\infty}\gamma^{t}u_{i}(\bar{\theta},{\bf {g}})\\
= & D+E+\sum_{t=0}^{K'-1}\gamma^{t}u_{i}(\bar{\theta},\hat{{\bf {g}}})+\gamma^{K'}(\sum_{t=0}^{\infty}\gamma^{t}(u_{i}(\bar{\theta},\hat{{\bf {g}}})-u_{i}(\bar{\theta},{\bf {g}})))\\
\leq & D+E+\sum_{t=0}^{K'-1}\gamma^{t}u_{i}(\bar{\theta},\hat{{\bf {g}}})+\frac{\gamma^{K'}F}{1-\gamma}\\
\leq & D+E+K'V+\frac{\gamma^{K'}F}{1-\gamma},
\end{align*}
where $F=\max_{\hat{I},\hat{{\bf {g}}}}\{u_{i}(\bar{\theta},\hat{{\bf {g}}})-u_{i}(\bar{\theta},{\bf {g}})\}$.
Since the total number of networks is finite, we know that $F$ exists
and that $F<0$.

Let $\gamma''$ be such that $\frac{F}{1-\gamma''}=-|D+E|-1$, and
let $\bar{K}$ and $\bar{\gamma}$ be as derived in the proof of Theorem
1. Let $\hat{K}=\bar{K}$ and let $\hat{\gamma}=\max\{\gamma'',\bar{\gamma}\}$.
For every $\gamma\geq\hat{\gamma}$, let $M(\gamma)$ be the largest
$K'\in\mathbb{N}$ such that $D+E+K'V+\frac{\gamma^{K'}F}{1-\gamma}<0$.
We know that $M(\gamma)$ exists because $K'=0$ always satisfies
the inequality.

Now, given $\hat{K}$ and any $\gamma\geq\hat{\gamma}$, $\hat{s}_{c}$
is an equilibrium where the formation process converges to ${\bf {g}}$
by Theorem 1. From the construction of $M(\gamma)$, given any $\hat{I}$
following any formation history with the remaining punishment phase
no longer than $M(\gamma)$ periods, there is always an agent in $\hat{I}\subsetneq I$
and associated $\hat{{\bf {g}}}$ whose payoff under strategy profile
$s'(\hat{I},\hat{{\bf {g}}})$ is strictly lower than that under strategy
profile $\hat{s}_{c}$. Hence, $\hat{s}_{c}$ is immune to $s'(\hat{I},\hat{{\bf {g}}})$.
Finally, since the term $D+E+K'V+\frac{\gamma^{K'}F}{1-\gamma}$ is
increasing in $K'$ and decreasing in $\gamma$, $M(\gamma)$ is increasing
in $\gamma$; the fact that $\lim_{\gamma\rightarrow1}\frac{\gamma^{K'}F}{1-\gamma}=-\infty$
for every given $K'$ ensures that $\lim_{\gamma\rightarrow1}M(\gamma)=\infty$.
This completes the proof. \end{proof}

To understand what this proposition means, note that $\hat{K}$ is
the length of the punishment period. In the equilibrium strategies
that we have constructed, agents receive a payoff of 0 during the
punishment phase. Hence, if $\hat{K}$ were infinite, or even extremely
long, groups would prefer to deviate rather than suffer such a long
punishment. However once $\hat{K}$ is given, part b guarantees that
if agents are sufficiently patient, groups of agents will be willing
to endure a punishment of length $\hat{K}$ rather than coordinate
on a network deviation.

\section{Foresight in the Connections Model}

We have shown that as long as a network secures a positive payoff
for every agent, it can be sustained in an equilibrium if the monitoring
structure is fine enough and agents are sufficiently patient. We now
apply this result and the techniques to evaluate the sustainability
of efficient entworks in the widely studied connections model introduced
by Jackson and Wolinsky\cite{JW}. Jackson and Wolinsky\cite{JW}
assume that agents are homogeneous and myopic and they argue that
the strongly efficient network is either empty, a star or a clique.
We allow for heteregenous and foresighted agents and find a much richer
set of strongly efficient networks.

Because the general case is cumbersome we first discuss in detail
a two-type environment to clearly explain the key results without
loss of much generality and to avoid technical redundancy. We will
demonstrate how the analysis can be extended to a generalized model
with multiple types at the end of this section.

\subsection{Characterization of Strongly Efficient Networks}

Assume that each agent can be one of two types, $\alpha$ or $\beta$.
Let $n_{\alpha},n_{\beta}>0$ denote the number of type $\alpha$
agents and that of type $\beta$ agents correspondingly. Without loss
of generality, we assume that $f(\alpha)>f(\beta)$. Let ${\bf {g}}^{e}$
denote the strongly efficient network. Before stating the formal result,
we first present a graphical illustration of the topology of the strongly
efficient network under different parameter values in Figure 1.

\begin{figure}[h]
\centering \includegraphics[width=5in]{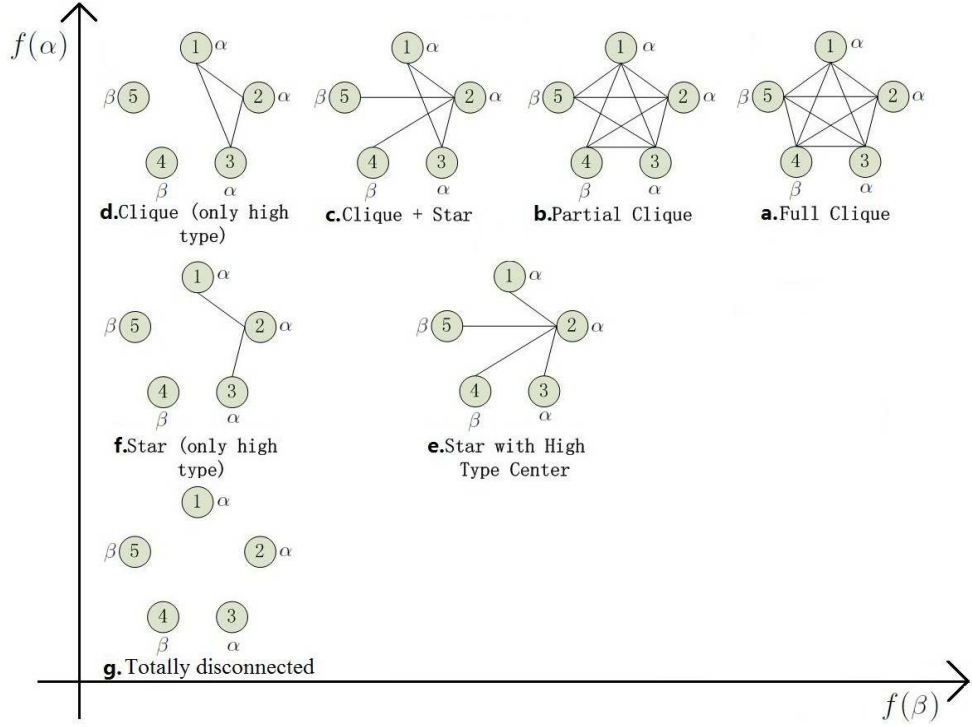}
\caption{Strongly efficient network in connections model}
\end{figure}

The following theorem fomally characerizes the conditions on model
parameters that lead to each strongly efficient network topology.

\begin{thm} ${\bf {g}}^{e}$ can be described as follows: 
\begin{itemize}
\item {a.} If $(1-\delta)f(\beta)>c$, then ${\bf {g}}^{e}$ is a clique
encompassing every agent. 
\item {b.} If $(1-\delta)\frac{f(\alpha)+f(\beta)}{2}>c>(1-\delta)f(\beta)$,
then ${\bf {g}}^{e}$ is such that every two type $\alpha$ agents
are linked, and every type $\alpha$ agent is linked with every type
$\beta$ agent, but no type $\beta$ agent is linked with another
type $\beta$ agent. 
\item {c.} If $(1-\delta)f(\alpha)>c>(1-\delta)\frac{f(\alpha)+f(\beta)}{2}$,
and $(1+\delta(n_{\alpha}-1))f(\alpha)+(1+\delta(n_{\alpha}+n_{\beta}-2))f(\beta)>2c$,
then ${\bf {g}}^{e}$ is such that every two type $\alpha$ agents
are linked, and every type $\beta$ agent is linked with the same
type $\alpha$ agent, but no type $\beta$ agent is linked with another
type $\beta$ agent. 
\item {d.} If $(1-\delta)f(\alpha)>c>(1-\delta)\frac{f(\alpha)+f(\beta)}{2}$,
and $(1+\delta(n_{\alpha}-1))f(\alpha)+(1+\delta(n_{\alpha}+n_{\beta}-2))f(\beta)<2c$,
then ${\bf {g}}^{e}$ is a clique encompassing every type $\alpha$
agent but no type $\beta$ agent. 
\item {e.} If $(1-\delta)f(\alpha)<c$, $f(\alpha)+f(\beta)+\delta[(n_{\beta}-1)f(\beta)+(n_{\alpha}-1)(f(\alpha)+f(\beta))]>2c$,
and 
\begin{align*}
 & 2(n_{\alpha}-1)f(\alpha)+n_{\beta}(f(\alpha)+f(\beta))+\delta[(n_{\alpha}-1)(n_{\alpha}-2)f(\alpha)+n_{\beta}(n_{\beta}-1)f(\beta)\\
 & +n_{\beta}(n_{\alpha}-1)(f(\alpha)+f(\beta))]-2(n_{\alpha}+n_{\beta}-1)c>0,
\end{align*}
then ${\bf {g}}^{e}$ is a star encompassing every agent, with a type
$\alpha$ agent as the center. 
\item {f.} If $(1-\delta)f(\alpha)<c$, and 
\begin{align*}
(1+\delta(n_{\alpha}-1))f(\alpha)+(1+\delta(n_{\alpha}+n_{\beta}-2))f(\beta)<2c<f(\alpha)(2+\delta(n_{\alpha}-2)),
\end{align*}
then ${\bf {g}}^{e}$ is a star encompassing every type $\alpha$
agent but no type $\beta$ agent. 
\item {g.} If $(1-\delta)f(\alpha)<c$, and 
\begin{align*}
 & \max\{2(n_{\alpha}-1)f(\alpha)+n_{\beta}(f(\alpha)+f(\beta))+\delta[(n_{\alpha}-1)(n_{\alpha}-2)f(\alpha)+n_{\beta}(n_{\beta}-1)f(\beta)\\
 & +n_{\beta}(n_{\alpha}-1)(f(\alpha)+f(\beta))]-2(n_{\alpha}+n_{\beta}-1)c,f(\alpha)(2+\delta(n_{\alpha}-2))-2c\}<0,
\end{align*}
then ${\bf {g}}^{e}$ is the empty network. 
\end{itemize}
\end{thm}

Despite the lengthy conditions on the payoffs from different types,
the underlying argument for the above characterization is straightforward.
According to how much benefit each type can provide via a connection,
we can categorize types in a systematic way and assign linkage correspondingly
to maximize the sum of total payoffs. For the high types among which
a direct link always brings benefits that are higher than the maintenance
cost, a clique must be formed among them in any strongly efficient
network. The next category contains the types for which any single
link to one of the highest types is socially beneficial, but links
among themselves are not. As a result, these types will not link directly
to themselves but will form every possible link to agents belonging
to the first category. When the benefit from a type gets even lower,
such types can only add to social welfare by having only one link
to one agent of the strictly highest type, thus minimizing the cost
and receiving/providing most of the benefit via indirect connection.
Last but not least, agents of the lowest types cannot increase the
social welfare in any way and will remain singletons in a strongly
efficient network.

The connected component of the strongly efficient network exhibits
a ``core-periphery'' pattern in topology. The core, which corresponds
to the first category, consists of agents with the highest connectivity
degree and the largest clustering coefficient. The periphery agents
each have one or more links with the core agents, depending on the
value structure. The detailed proof of the theorem is provided below.

\begin{proof} $(a)$ The result is clear since $(1-\delta)f(\beta)>c$
implies that the benefit from any link (bounded below by $2(1-\delta)f(\beta)$)
is greater than the associated cost ($2c$).

$(b)$ If $(1-\delta)\frac{f(\alpha)+f(\beta)}{2}>c$, a link between
two type $\alpha$ agents or one type $\alpha$ agent and one type
$\beta$ agent always increases the total payoff in the network. Given
that there is a link between these agents, $c>(1-\delta)f(\beta)$
implies that a link between two type $\beta$ agents always decreases
the total payoff in the network. Hence, ${\bf {g}}^{e}$ is as described
in the result.

$(c)$ The first condition implies that any pair of type $\alpha$
agents are linked in ${\bf {g}}^{e}$. Furthermore, for $n$ type
$\beta$ agents with $m_{1},m_{2},\cdots,m_{n}$ links respectively
(where $m_{1},m_{2},\cdots,m_{n}$ are positive integers), the largest
possible contribution to total payoff is 
\begin{align*}
\sum_{k=1}^{n}[m_{k}(f(\alpha)+f(\beta))+\delta((n_{\alpha}-m_{k})(f(\alpha)+f(\beta))+(n_{\beta}-1)2f(\beta))-2m_{k}c].
\end{align*}
Given the condition $c>(1-\delta)\frac{f(\alpha)+f(\beta)}{2}$, the
above value is maximized at $m_{k}=1$ for $k=1,2,\cdots,n$. This
upper bound is reached when all $n$ type $\beta$ agents are linked
to the same type $\alpha$ agent. It is not difficult to see that
if the contribution by $n$ type $\beta$ agents is positive, then
that by $n+1$ connected type $\beta$ agents is also positive and
larger. Hence, in ${\bf {g}}^{e}$, either no type $\beta$ agent
is connected, or every type $\beta$ agent is linked to the same type
$\alpha$ agent. The condition $(1+\delta(n_{\alpha}-1))f(\alpha)+(1+\delta(n_{\alpha}+n_{\beta}-2))f(\beta)>2c$
implies that the contribution by $n_{\beta}$ type $\beta$ agents
is positive, and hence ${\bf {g}}^{e}$ is as described in the result.

$(d)$ It follows from the proof of $(c)$.

$(e)$ First, we establish the following lemma. \begin{lem} Any strongly
efficient network has at most one non-empty component. \end{lem}
\begin{proof} Suppose that there exists some strongly efficient network
that has two non-empty components $C_{1}$ and $C_{2}$. For component
$m$ ($m=1,2$), let $i_{m}$ denote (one of) the agent(s) who has
the highest payoff in component $C_{m}$. Since the network is strongly
efficient, we know that this payoff is non-negative for $m=1,2$.
Let $D_{m}$ be the set of links that $i_{m}$ has in component $C_{m}$.
for every $i_{m}j\in D_{m}$ and every $D_{m}'\subset D_{m}-i_{m}j$,
define 
\begin{align*}
\Delta u_{m}(i_{m}j,D'_{m})=u_{{i_{m}}}(\bar{\theta}_{-{i_{m}}},(C_{m}\setminus D_{m})\cup D_{m}')-u_{{i_{m}}}(\bar{\theta}_{-{i_{m}}},(C_{m}\setminus D_{m})\cup D_{m}'-i_{m}j).
\end{align*}
This term denotes the marginal payoff of $i$ from link $i_{m}j$
given $D_{m}'$. It is not difficult to see that for $D''_{m}\subset D'_{m}$,
$\Delta u_{m}(i_{m}j,D''_{m})\geq\Delta u_{m}(i_{m}j,D'_{m})$. Since
agent $i_{m}$'s payoff is non-negative, it follows that there must
exist $j_{m}$ such that $\Delta u_{m}(i_{m}j,\varnothing)\geq0$.

Consider the network $C_{1}\cup C_{2}+j_{1}j_{2}$. For $j_{1}$,
the marginal payoff from link $j_{1}j_{2}$ is strictly larger than
$\Delta u_{2}(i_{2}j_{2},\varnothing)$; similarly for $j_{2}$, the
marginal payoff from link $j_{1}j_{2}$ is strictly larger than $\Delta u_{1}(i_{1}j_{1},\varnothing)$.
Hence, the total payoff in this network is strictly higher than that
in the original network $C_{1}\cup C_{2}$, which contradicts the
assumption that $C_{1}\cup C_{2}$ is strongly efficient. \end{proof}

We then show that if $(1-\delta)f(\alpha)<c$, for every non-empty
component ${\bf {g}}$, the total payoff in this component is weakly
less than that in a star component with a type $\alpha$ agent as
the center, denoted ${\bf {g}}^{*}$. Let $k_{\alpha\alpha}^{d},k_{\alpha\beta}^{d},k_{\beta\beta}^{d}$
denote the number of links between two type $\alpha$ agents, one
type $\alpha$ agent and one type $\beta$ agent, and two type $\beta$
agents respectively. Let $k_{\alpha\alpha}^{ind},k_{\alpha\beta}^{ind},k_{\beta\beta}^{ind}$
denote the number of shortest indirect paths between two agents in
the same three cases above. Let $\nu({\bf {g}})$ and $\nu({\bf {g}}^{*})$
denote the total payoff in the original component and the star component
respectively. Note that the length of any indirect path is at least
$2$, and thus we have 
\begin{align*}
\nu({\bf {g}})\leq & (k_{\alpha\alpha}^{d}+\delta k_{\alpha\alpha}^{ind})2f(\alpha)+(k_{\alpha\beta}^{d}+\delta k_{\alpha\beta}^{ind})(f(\alpha)+f(\beta))+(k_{\beta\beta}^{d}+\delta k_{\beta\beta}^{ind})2f(\beta)\\
 & -2(k_{\alpha\alpha}^{d}+k_{\alpha\beta}^{d}+k_{\beta\beta}^{d})c\\
\nu({\bf {g}}^{*})= & (n_{\alpha}-1+\delta\frac{(n_{\alpha}-1)(n_{\alpha}-2)}{2})2f(\alpha)+(n_{\beta}+\delta n_{\beta}(n_{\alpha}-1))(f(\alpha)+f(\beta))\\
 & +\delta\frac{n_{\beta}(n_{\beta}-1)}{2}2f(\beta)-2(n_{\alpha}+n_{\beta}-1)c.
\end{align*}
Note that 
\begin{align*}
k_{\alpha\alpha}^{d}+k_{\alpha\alpha}^{ind} & =\frac{n_{\alpha}(n_{\alpha}-1)}{2}\\
k_{\alpha\beta}^{d}+k_{\alpha\beta}^{ind} & =n_{\alpha}n_{\beta}\\
k_{\beta\beta}^{d}+k_{\beta\beta}^{ind} & =\frac{n_{\beta}(n_{\beta}-1)}{2}.
\end{align*}
Then we have 
\begin{align*}
\nu({\bf {g}})-\nu({\bf {g}}^{*})\leq & (1-\delta)((k_{\alpha\alpha}^{d}-(n_{\alpha}-1))2f(\alpha)+(k_{\alpha\beta}^{d}-n_{\beta})(f(\alpha)+f(\beta))+k_{\beta\beta}^{d}2f(\beta))\\
 & -2(k_{\alpha\alpha}^{d}+k_{\alpha\beta}^{d}+k_{\beta\beta}^{d}-(n_{\alpha}+n_{\beta}-1))c\\
\leq & 2(k_{\alpha\alpha}^{d}+k_{\alpha\beta}^{d}+k_{\beta\beta}^{d}-(n_{\alpha}+n_{\beta}-1))((1-\delta)f(\alpha)-c)\\
\leq & 0.
\end{align*}
Finally, note that equality is achieved if and only if (1) $k_{\alpha\beta}^{d}=k_{\beta}$,
$k_{\beta\beta}^{d}=0$, $k_{\alpha\alpha}^{d}+k_{\alpha\beta}^{d}+k_{\beta\beta}^{d}=n_{\alpha}+n_{\beta}-1$
and (2) there is no (shortest) indirect path with length greater than
two between any two type $\alpha$ agents. These two conditions are
satisfied if and only if ${\bf {g}}$ is also a star component with
a type $\alpha$ agent as the center.

It is clear that if a star component with a type $\alpha$ agent as
the center results in a positive total payoff, then adding an agent
of type $\alpha$ as the periphery will increase the total payoff;
moreover, if the marginal payoff brought about by all the type $\beta$
agents in the star component is positive, then adding an agent of
type $\beta$ as the periphery will increase the total payoff. Hence,
${\bf {g}}^{e}$ can only be one of the following three: (1) the empty
network, (2) a star encompassing only type $\alpha$ agents, and (3)
a star encompassing all agents, with a type $\alpha$ agent as the
center. The second and the third conditions in the result ensure that
the marginal payoff brought about by all the type $\beta$ agents
is positive, and that (3) has a positive total payoff. Hence, (3)
is ${\bf {g}}^{e}$.

$(f)$ It follows from the proof of $(e)$.

$(g)$ It follows from the proof of $(e)$. \end{proof}

\subsection{Strong Efficiency and Core-Stability}

The connections model also allows us to investigate the relation between
a strongly efficient network and a core-stable network. Note that
these two concepts do not imply one another. In a strongly efficient
network, there may be a subgroup of agents that can improve the payoff
of each member by ``local autonomy'': for instance, the center of
a star will be strictly better off staying a singleton if all periphery
agents are of low type. On the other hand, even though the agents
cannot improve everyone's payoff in a core-stable network, there may
be a way to improve the total payoff. An example would be to switch
the center of a star from a low type agent to a high type one.

In this section, we establish necessary and sufficient conditions
for a strongly efficient network to be core-stable. It is helpful
to first compute the largest possible one-period payoff an agent of
each type can get in any network.

\begin{prop} For $\theta\in{\alpha,\beta}$ let $V(\theta)$ denote
the maximum payoff that an agent of type $\theta$ can obtain in any
network in a single period. We have: 
\begin{itemize}
\item {a.} If $(1-\delta)f(\beta)>c$, then 
\begin{align*}
V(\alpha) & =(n_{\alpha}-1)f(\alpha)+n_{\beta}f(\beta)-(n_{\alpha}+n_{\beta}-1)c\\
V(\beta) & =n_{\alpha}f(\alpha)+(n_{\beta}-1)f(\beta)-(n_{\alpha}+n_{\beta}-1)c.
\end{align*}

\item {b.} If $(1-\delta)f(\alpha)>c>(1-\delta)f(\beta)$, then 
\begin{align*}
V(\alpha) & =(n_{\alpha}-1)f(\alpha)+\delta n_{\beta}f(\beta)-(n_{\alpha}-1)c\\
V(\beta) & =n_{\alpha}f(\alpha)+\delta(n_{\beta}-1)f(\beta)-n_{\alpha}c.
\end{align*}

\item {c.} If $\min\{f(\alpha)+\delta((n_{\alpha}-2)f(\alpha)+n_{\beta}f(\beta)),f(\alpha)+\delta((n_{\alpha}-1)f(\alpha)+(n_{\beta}-1)f(\beta))\}>c>(1-\delta)f(\alpha)$,
then 
\begin{align*}
V(\alpha) & =f(\alpha)+\delta((n_{\alpha}-2)f(\alpha)+n_{\beta}f(\beta))-c\\
V(\beta) & =f(\alpha)+\delta((n_{\alpha}-1)f(\alpha)+(n_{\beta}-1)f(\beta))-c.
\end{align*}

\item {d.} If $f(\alpha)+\delta((n_{\alpha}-1)f(\alpha)+(n_{\beta}-1)f(\beta))>c>\max\{f(\alpha)+\delta((n_{\alpha}-2)f(\alpha)+n_{\beta}f(\beta)),(1-\delta)f(\alpha)\}$,
then 
\begin{align*}
V(\alpha) & =0\\
V(\beta) & =f(\alpha)+\delta((n_{\alpha}-1)f(\alpha)+(n_{\beta}-1)f(\beta))-c.
\end{align*}

\item {e.} If $\max\{f(\alpha)+\delta((n_{\alpha}-1)f(\alpha)+(n_{\beta}-1)f(\beta)),f(\alpha)+\delta((n_{\alpha}-2)f(\alpha)+n_{\beta}f(\beta)),(1-\delta)f(\alpha)\}<c$,
then 
\begin{align*}
V(\alpha) & =0\\
V(\beta) & =0.
\end{align*}

\end{itemize}
\end{prop}

\begin{proof} $(a)$ Let $\hat{V}(\theta,m)$ be the largest possible
payoff of an agent of type $\theta$ with $m$ links. We have

\begin{center}
$\hat{V}(\alpha,m)=\left\{ \begin{array}{c}
0,\text{ if \ensuremath{m=0}}\\
mf(\alpha)+\delta((n_{\alpha}-1-m)f(\alpha)+n_{\beta}f(\beta))-mc,\text{ if \ensuremath{0<m\leq n_{\alpha}-1}}\\
(n_{\alpha}-1)f(\alpha)+(m-(n_{\alpha}-1))f(\beta)+\delta(n_{\alpha}+n_{\beta}-1-m)f(\beta)-mc,\text{ otherwise}
\end{array}\right.$ $\hat{V}(\beta,m)=\left\{ \begin{array}{c}
0,\text{ if \ensuremath{m=0}}\\
mf(\alpha)+\delta((n_{\alpha}-m)f(\alpha)+n_{\beta}f(\beta))-mc,\text{ if \ensuremath{m\leq n_{\alpha}}}\\
n_{\alpha}f(\alpha)+(m-n_{\alpha})f(\beta)+\delta(n_{\alpha}+n_{\beta}-1-m)f(\beta)-mc,\text{ otherwise.}
\end{array}\right.$ 
\par\end{center}

It is clear that these largest possible payoffs are achievable (for
instance, let the agent be a periphery agent in a star with a type
$\alpha$ agent as the center, and form the rest of the $m$ links
first with type $\alpha$ agents, then with type $\beta$ agents if
she is already linked with every type $\alpha$ agent.). If $(1-\delta)f(\beta)>c$,
for an agent of either type her payoff is maximized when she makes
every possible link, hence $V(\alpha)$ and $V(\beta)$ are as shown
in the result.

$(b)$ If $(1-\delta)f(\alpha)>c>(1-\delta)f(\beta)$, for an agent
of either type her payoff is maximized when she links with every type
$\alpha$ agent but with no type $\beta$ agent, hence $V(\alpha)$
and $V(\beta)$ are as shown in the result.

$(c)$ If $c>(1-\delta)f(\alpha)$, given that an agent is connected
(not a singleton), her largest payoff is higher when she has fewer
links. Hence, her payoff is $\hat{V}(\theta,1)$ if $\hat{V}(\theta,1)\geq0$
and $0$ otherwise. Hence $V(\alpha)$ and $V(\beta)$ are as shown
in the result.

$(d)$ It follows from the proof of $(c)$.

$(e)$ It follows from the proof of $(c)$. \end{proof}

We can see from the proof above that the largest payoff an agent obtains
from a network is closely related to the network topology. Hence,
if a network offers the largest possible payoff to most of the agents,
it is very likely to be core-stable since those agents' payoffs cannot
be improved further. The final criterion of stability then rests on
whether the few agents that do not get the highest possible payoff
can form a beneficial coalition. Using this argument, we inspect the
strongly efficient network in every possible type vector and present
the result below.

\begin{prop} Consider cases ($a$) -- ($g$) in Theorem 2. We have: 
\begin{itemize}
\item {a, d, g.} ${\bf {g}}^{e}$ is core-stable. 
\item {b.} ${\bf {g}}^{e}$ is core-stable if and only if $f(\beta)\geq c$. 
\item {c.} ${\bf {g}}^{e}$ is core-stable if and only if a network that
links one type $\alpha$ agent to all the type $\beta$ agents yields
the $\alpha$ agent a non-negative payoff. 
\item {e.} ${\bf {g}}^{e}$ is core-stable if and only if a network that
has a type $\alpha$ agent at the center yields the $\alpha$ agent
a non-negative payoff. 
\item {f.} ${\bf {g}}^{e}$ is core-stable if and only if $f(\alpha)\geq c$. 
\end{itemize}
\end{prop}

\begin{proof} $(a)(d)(g)$ The cases $(a)$ and $(g)$ are clear.
For $(d)$, suppose that ${\bf {g}}^{e}$ is not core-stable, it then
follows that any blocking group $I'$ (with network ${\bf {g}}'$)
must contain at least one type $\alpha$ agent, but not all type $\alpha$
agents. Consider the network ${\bf {g}}'$ that blocks ${\bf {g}}^{e}$.
Since every agent in $I'$ has a weakly higher payoff in ${\bf {g}}'$
than in ${\bf {g}}^{e}$ and some agent in $I'$ has a strictly higher
payoff in ${\bf {g}}'$ than in ${\bf {g}}^{e}$, the total payoff
of agents in $I'$ must be strictly higher than that in ${\bf {g}}^{e}$.
By Theorem 2, we can re-organize ${\bf {g}}'$ into a clique with
only type $\alpha$ agents to yield an even higher total payoff. However,
if a clique with only type $\alpha$ agents has a positive total payoff,
then it is impossible for a proper subset of these agents to form
a clique with a higher total payoff than their total payoff in the
original clique. Hence we have a contradiction.

$(b)$ If $f(\beta)<c$, then ${\bf {g}}^{e}$ is not core-stable
because a clique formed by all type $\alpha$ agents blocks ${\bf {g}}^{e}$.
Suppose that $f(\beta)\geq c$ and that ${\bf {g}}^{e}$ is not core-stable,
then any blocking group $I'$ either only contains type $\alpha$
agents, or contains all the agents because each type $\beta$ agent
in ${\bf {g}}^{e}$ gets payoff $V(\beta)$, and getting $V(\beta)$
requires connection to every other agent. Both cases contradict the
fact that ${\bf {g}}^{e}$ is strongly efficient.

$(c)$ Let $i$ denote the type $\alpha$ agent linking with all the
type $\beta$ agents. If $i$ has a negative payoff, then ${\bf {g}}^{e}$
is not core-stable because $\{i\}$ blocks ${\bf {g}}^{e}$. Suppose
that $i$ has a non-negative payoff and that ${\bf {g}}^{e}$ is not
core-stable, then any blocking group $I'$ either only contains no
type $\alpha$ agent other than $i$ and at least one type $\beta$
agent, or contains all the agents because each type $\alpha$ agent
other than $i$ in ${\bf {g}}^{e}$ gets payoff $V(\alpha)$, and
getting $V(\alpha)$ requires connection to every other agent. The
second case contradicts the fact that ${\bf {g}}^{e}$ is strongly
efficient. In the first case, note that the largest payoff of any
type $\beta$ agent in $I'$ is strictly less than $\hat{V}(\beta,1)$
(since $c>(1-\delta)f(\beta)$), which is the payoff of every type
$\beta$ agent in ${\bf {g}}^{e}$. Hence we again have a contradiction.

$(e)$ Let $j$ denote the type $\alpha$ agent at the center. If
$j$ has a negative payoff, then ${\bf {g}}^{e}$ is not core-stable
because $\{j\}$ blocks ${\bf {g}}^{e}$. Suppose that $j$ has a
non-negative payoff and that ${\bf {g}}^{e}$ is not core-stable,
then any blocking group $I'$ either only contains $j$, or contains
all the agents because each agent other than $j$ in ${\bf {g}}^{e}$
gets payoff $V(\alpha)$ (or $V(\beta)$), and getting $V(\alpha)$
(or $V(\beta)$) requires connection to every other agent. The first
case contradicts the assumptions that $I'$ blocks ${\bf {g}}^{e}$
and $j$ has a non-negative payoff in ${\bf {g}}^{e}$, and the second
case contradicts the fact that ${\bf {g}}^{e}$ is strongly efficient.

$(f)$ If $f(\alpha)<c$, then ${\bf {g}}^{e}$ is not core-stable
because the center agent in ${\bf {g}}^{e}$ blocks ${\bf {g}}^{e}$.
Suppose that $f(\alpha)\geq c$ and that ${\bf {g}}^{e}$ is not core-stable,
it then follows that any blocking group $I'$ (with network ${\bf {g}}'$)
must contain at least one type $\alpha$ agent, but not all type $\alpha$
agents. Consider the network ${\bf {g}}'$ that blocks ${\bf {g}}^{e}$.
Since every agent in $I'$ has a weakly higher payoff in ${\bf {g}}'$
than in ${\bf {g}}^{e}$ and some agent in $I'$ has a strictly higher
payoff in ${\bf {g}}'$ than in ${\bf {g}}^{e}$, the total payoff
of agents in $I'$ must be strictly higher than that in ${\bf {g}}^{e}$.
By Proposition 2, we can re-organize ${\bf {g}}'$ into a star with
only type $\alpha$ agents to yield an even higher total payoff. However,
if $f(\alpha)\geq c$, it is impossible for a proper subset of the
type $\alpha$ agents to form a star with a higher total payoff than
their total payoff in ${\bf {g}}^{e}$. Hence we have a contradiction.
\end{proof}

The key to whether a strongly efficient network is core-stable is
its ``weakest link'': the agent maintaining the most links but enjoying
the smallest net payoff. If this particular agent gets a negative
payoff, she would prefer to simply sever all her links and remain
a singleton.

\subsection{Effect of Spatial Discount}

A crucial factor in the connections model is the spatial discount
factor $\delta$. It determines the rate of payoff depreciation as
two agents become further apart in connection, and hence influences
an agent's incentives of directly linking to another already connected
agent. As a result, a change in $\delta$ affects \textit{every} connected
agent's payoff and has a non-negligible impact on the set of sustainable
networks in equilibrium. The importance of the spatial discount in
a static environment and a dynamic one with myopia has been noted
in Jackson and Wolinsky\cite{JW} and Song and van der Schaar\cite{SV}.
In the dynamics with foresight, the role of $\delta$ becomes even
more important than the previous cases because it enters the payoff
an agent obtains for every period. As it turns out, the set of sustainable
networks changes monotonically with the value of $\delta$.

Given $\bar{\theta}$, let $G(\delta)$ denote the set of networks
${\bf {g}}$ for which there is a cutoff $\bar{\gamma}\in(0,1)$ such
that if $\gamma\in[\bar{\gamma},1)$, then there exists an equilibrium
in which the formation process converges strongly to ${\bf {g}}$.

\begin{prop} If $\delta_{1}>\delta_{2}$, then $G(\delta_{1})\supset G(\delta_{2})$.
\end{prop}

\begin{proof} For every ${\bf {g}}\in G(\delta_{2})$, it follows
from Theorem $1(a)$ that every connected agent in ${\bf {g}}$ gets
a non-negative payoff when $\delta=\delta_{2}$. Since $\delta_{1}>\delta_{2}$,
every connected agent in ${\bf {g}}$ gets a positive payoff when
$\delta=\delta_{1}$. Then by Theorem $1$ we know that ${\bf {g}}\in G(\delta_{1})$.
\end{proof}

Although the set of networks that are sustainable in equilibrium is
monotone in $\delta$, the strongly efficient networks that are sustainable
in equilibrium is not monotone in $\delta$. Consider the following
example: $f(\alpha)=3$, $f(\beta)=0.5$, $c=2$, $n_{\alpha}=1$
and $n_{\beta}=3$. From the previous analysis, we know that the strongly
efficient network is either a star with the type $\alpha$ agent as
the center, or empty. On one hand, when $\delta$ is low (for instance
$\delta=0.1$) the efficient network is empty and clearly can be supported
in equilibrium. On the other hand, when $\delta$ is high (for instance
$\delta=0.8$) the efficient network is a star. However, it is clear
that a star network cannot be supported in equilibrium since the center
agent has a negative payoff. Even though a higher $\delta$ sustains
a larger set of networks, the strongly efficient network changes with
$\delta$ as well.

The spatial discount also partially determines the required patience
level of agents to sustain a network in equilibrium. When $\delta$
changes in a way that forming/maintaining links becomes more beneficial
to oneself in every period, every connected agent in a network has
higher incentive, and at the same time needs less patience, to sustain
that network in the long run. We demonstrate this by using our constructed
equilibrium $\hat{s}_{c}$. Given $\bar{\theta}$, $K$ and ${\bf {g}}$,
let $\gamma(\delta)$ denote the smallest $\gamma$ such that $\hat{s}_{c}$
is an equilibrium (if such $\gamma$ exists) under spatial discount
$\delta$.

\begin{prop} There exists $\hat{\delta}\leq1$ such that for every
$\delta_{1}$, $\delta_{2}\in[0,\hat{\delta}]$ such that $\delta_{1}>\delta_{2}$
and both $\gamma(\delta_{1})$ and $\gamma(\delta_{2})$ exist, $\gamma(\delta_{1})<\gamma(\delta_{2})$.
\end{prop}

\begin{proof} Note that the only incentive problem faced by an agent
is whether to sever or refuse to form a link in ${\bf {g}}$ during
the cooperation phase. Doing so saves the agent the cost $c$, and
the agent's loss in benefit has the form of 
\begin{align*}
\sum_{l=1}^{L}(\delta^{a_{l}}-\delta^{b_{l}})f(\theta_{j_{l}})
\end{align*}
where $a_{l}\in\mathbb{N}$, $b_{l}\in\mathbb{N}\cup\{\infty\}$ satisfy
$a_{l}<b_{l}$. When $b_{l}=\infty$, clearly $\delta^{a_{l}}-\delta^{b_{l}}$
is increasing in $\delta$. When $b_{l}\in\mathbb{N}$, we have 
\begin{align*}
\frac{d(\delta^{a_{l}}-\delta^{b_{l}})}{d\delta}=a_{l}\delta^{a_{l}-1}-b_{l}\delta^{b_{l}-1}=\delta^{a_{l}-1}(a_{l}-b_{l}\delta^{b_{l}-a_{l}})
\end{align*}
This derivative is positive when $\delta$ is sufficiently small.
In other words, the benefit from an existing link is increasing in
$\delta$ when $\delta$ is sufficiently small. Let $\hat{\delta}\leq1$
be the largest $\delta$ such that every $\delta\in[0,\hat{\delta}]$
satisfies this property.

Consider $\delta_{1}$, $\delta_{2}\in[0,\hat{\delta}]$ such that
$\delta_{1}>\delta_{2}$ and both $\gamma(\delta_{1})$ and $\gamma(\delta_{2})$
exist, and $\gamma\geq\gamma(\delta_{2})$. Note that by the above
property, in any period during the cooperation phase, the benefit
from following the equilibrium strategy in the current period and
in every subsequent period is increasing in $\delta$, and the cost
stays the same. Hence, if $\hat{s}_{c}^{y_{c}}$ is an equilibrium
given $\delta_{2}$ and $\gamma$, it must also be an equilibrium
given $\delta_{1}$ and $\gamma$. We can then conclude that $\gamma(\delta_{1})<\gamma(\delta_{2})$.
\end{proof}

\subsection{Generalized Connections Model}

In this section, we generalize the above two-type connections model
to a much richer multi-type environment. Consider a type set $\Theta$
with finitely many types, and let the payoff of agent $i$ from connecting
to agent $j$ be $f(\theta_{j})$ before spatial discount. Our results
on the unique efficient network and stability can be easily extended
to this model. The unique efficient network ${\bf {g}}^{e}$ has a
general \textit{core-periphery} structure, characterized by a partition
of agents $\{I_{1},I_{2},I_{3},I_{4}\}$ induced by a triplet of types
$\{\theta^{1},\theta^{2},\theta^{3}\}$ such that $f(\theta^{1})\geq f(\theta^{2})\geq f(\theta^{3})$
: 
\begin{itemize}
\item {a.} \textbf{Core} $I_{1}=\{i:f(\theta_{i})\geq f(\theta^{1})\}$:
the maximal subset of inter-linked agents. 
\item {b.} \textbf{Periphery I} $I_{2}=\{i:f(\theta^{2})\leq f(\theta_{i})<f(\theta^{1})\}$:
the maximal subset of agents that do not belong to $I_{1}$ but link
to every agent of some type(s) in $I_{1}$. 
\item {c.} \textbf{Periphery II} $I_{3}=\{i:f(\theta^{3})\leq f(\theta_{i})<f(\theta^{2})\}$:
the maximal subset of agents that do not belong to $I_{1}\cup I_{2}$
but link to the same agent in $I_{1}$. 
\item {d.} \textbf{Singleton} $I_{4}=\{i:f(\theta_{i})<f(\theta^{3})\}$:
the maximal subset of agents that do not link to any other agent. 
\end{itemize}
In the unique efficient network all agents in the core are linked
to each other, all agents in periphery I are linked to all agents
of at least one type in the core, all agents in periphery II are linked
and only linked to the same agent in the core, and agents that are
singletons are unlinked. Figure 2 below illustrates a sample core-periphery
network, with one type in each category.

\begin{figure}[h]
\centering \includegraphics[width=5in]{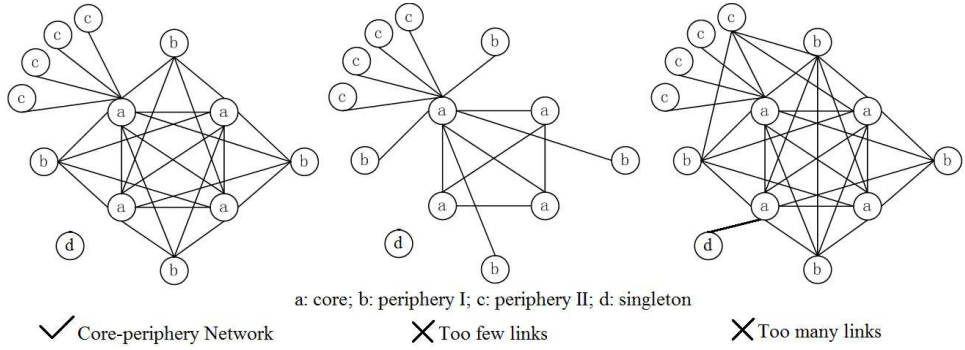}
\caption{Sample core-periphery network}
\end{figure}

The type cutoffs, $\{\theta^{1},\theta^{2},\theta^{3}\}$, are determined
in a similar way to the two-type model, only with more tedious calculations
on an agent's maximum possible contribution to the total payoff of
the group. A rough intuition for this result is that an agent of a
higher type takes more responsibility, in the sense that she should
form more links to create more value for the greater good. The overall
distinguishing features of the efficient network, especially when
the number of agents gets large, include a small diameter, a large
ratio between number of links and number of agents, and a large ratio
of number of ``triangles'' (connected triples of agents) and number
of agents. We show in the next section that in realistic situations
where the connections model applies, the level of coordination among
individuals is significantly higher than predicted by previous theories,
and that our model with foresight can account for a considerable proportion
of cooperation behavior in these endogenously formed networks.

Next, we can derive a necessary and sufficient condition for the efficient
network to be core-stable: \textit{an efficient network is core-stable
if and only if the agent(s) of the highest type has a non-negative
payoff}. The argument underlying the ``if'' part of this result
(the ``only if'' part is clear) is similar to the proof of Proposition
5. Suppose that there is a blocking group $I'$ and a corresponding
network ${\bf {g}}'$ on $I'$ that yields a weakly higher payoff
for every agent in $I'$ and a strictly higher payoff for at least
one agent in $I'$. We can always re-organize ${\bf {g}}'$ according
to Theorem 2 to obtain a weakly higher total payoff for the group
$I'$; then the agent with the highest payoff in the new network must
be at least as better-off as she is in ${\bf {g}}^{e}$. This agent
cannot belong to $I_{2}$ or $I_{3}$ since the agents in these two
categories in ${\bf {g}}^{e}$ have already enjoyed the unique highest
possible payoff that can only be provided by ${\bf {g}}^{e}$; however,
if this agent belongs to $I_{1}$, then it must be the case that some
other agent in $I_{2}$ gets a lower payoff than in ${\bf {g}}^{e}$,
a contradiction. An important message conveyed by this result is that
for efficiency to be achieved in equilibrium and to prevent coalitional
deviation, it is sufficient to focus on the agents of the highest
type, ensuring that their cost of maintaining links are covered by
the benefit from connection.

\subsection{Empirical Comparisons}

We have provided a full characterization of the strongly efficient
network in a standard connections model with\textit{ heterogeneous}
agents, in which we find that such networks generally exhibit a ``core-periphery''
structure. Prominent features of such network topologies in terms
of several descriptive statistics are: (1) a large average local clustering
coefficient (ALCC), measured by the number of pairs of linked neighbors
devided by the number of possible pairs of neighbors; (2) a large
global clustering coefficient, measured by the number of closed triangles
devided by the number of triangles; (3) a short diameter (D), measured
by the number of links in the longest of all shortest paths between
any two agents). Both clustering coefficients indicate the degree
to which small groups of agents in a network tend to keep close ties
to each other, and the diameter is an index of the entire network's
density. Note that a large clustering coefficient does \textit{not}
guarantee a small diameter. For instance, a ``chain'' network created
by connecting many small cliques has a large clustering coefficient
but a large diameter.

Our findings are consistent with data collected from existing real
networks. To see this, we compare our predictions with data obtained
from sample social networks from Facebook and collaboration networks
of Arxiv High Energy Physics (AHEP) \footnote{Source of datasets: SNAP Datasets: Stanford Large Network Dataset
Collection \cite{snapnets}.}, and also with simulated networks following models with myopic agents,
using pairwise stability introduced by Jackson and Wolinsky\cite{JW}
as the solution concept in each period. In the simulation ``Myopic
1'', we assume that the payoff from connecting to any one agent before
spatial discount is 10, the link maintenance cost is 5, and the spatial
discount factor is 0.6. In the simulation ``Myopic 2'', we assume
that there are three types of agents; connecting to each type yields
a payoff of 16, 10 and 6 before spatial discount respectively, and
the cost and the spatial discount factor remain the same. The ratio
of types is $1:2:3$, that is type 1, 2 and 3 agents account for $\frac{1}{6}$,
$\frac{1}{3}$ and $\frac{1}{2}$ of the population correspondingly.
The simulated network formation process is run for a sufficiently
long time such that each pair of agents is selected at least twice
in expectation. The ``Foresighted Model'' column shows descriptive
statistics for the strongly efficient network in our model, with the
same group size and type distribution as ``Myopic 2''. Table 1 below
provides summary statistics on the networks and Figure 3 illustrates
the actual network topology in AHEP\footnote{This visualization of network is provided by Tim Davis at TAMU. Retrieved
from http://www.cise.ufl.edu/research/sparse/matrices/SNAP/.}. In Table 1, the entry ``90\% D'' represents the 90th percentile
in the distribution of path length.

We find that the actual networks recorded are considerably closer
to those predicted by our model with foresighted agents than by models
with myopic agents which are representative of much previous literature.
In the two models with myopia, the network is not clustered (small
ALCC and GCC), suggesting a relatively small group of ``super star''
agents that link to many ``periphery'' agents, but showing little
direct relation among the ``periphery'' agents. This is not true
for the actual networks (large ALCC and GCC). In contrast, the strongly
efficient network we have characterized captures this charateristics
of high clustering, and we have shown that when agents are foresighted
this network can be supported in equilibrium. It corroborates our
earlier statement that our model with foresight provides a more appropriate
framework of analyzing network formation, which leads to more realistic
predictions for actual networks.

Another observation that can be made based on these results is that
the formed networks are rather dense (small D), confirming the well-studied
“small world” phenomenon. However, the diameter of an actual network
is typically larger than that predicted by the network formations
models. We believe that this difference in diameter results from the
simplistic meeting process adopted in all this literature: individuals
in an actual network do not meet with uniform probabilities; instead,
some agents may meet more often while others only seldom. Hence, we
believe that an important topic of future research is understanding
the effect of different meeting processes on the emerging networks.
\begin{table}
\centering%
\begin{tabular}{|c|c|c|c|c|c|}
\hline 
\multirow{2}{*}{} & \multicolumn{2}{c|}{Actual} & \multicolumn{2}{c|}{Simulation} & \multirow{2}{*}{Foresighted Model}\tabularnewline
\cline{2-5} 
 & Facebook  & AHEP  & Myopic 1  & Myopic 2  & \tabularnewline
\hline 
\hline 
$N$  & 4039  & 12008  & 1000  & 1000  & 1000\tabularnewline
\hline 
ALCC  & 0.6055  & 0.6115  & 0.1357  & 0.1957  & 0.4251\tabularnewline
\hline 
GCC  & 0.2647  & 0.3923  & 0.0458  & 0.0756  & 0.3570\tabularnewline
\hline 
D  & 8  & 13  & 2 & 2 & 2\tabularnewline
\hline 
90\% D  & 4.7  & 5.3  & 2  & 2  & 2\tabularnewline
\hline 
\end{tabular}

\caption{Summary statistics of networks}
\end{table}

\begin{figure}[h]
\centering \includegraphics[width=4in]{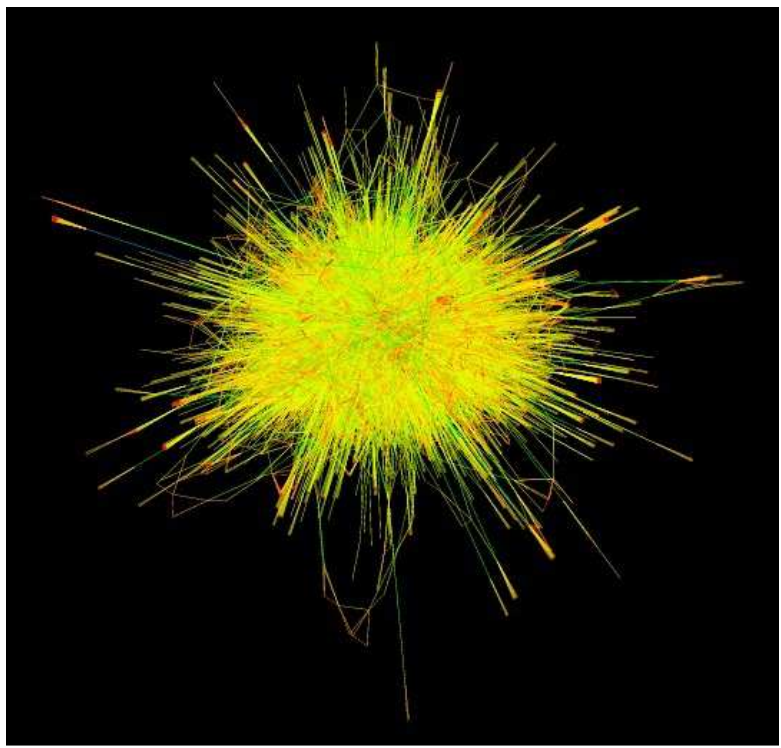}
\caption{AHEP network}
\end{figure}

\section{Network Convergence Theorem with Incomplete Information}

In real-life applications, agents will not usually know the types
of the other agents before they are linked to them. As we have shown
in our prior work in\cite{SV}, the introduction of incomplete information
leads to significant differences in agents' strategic behavior and
equilibrium network topology. In this section, we extend the Network
Convergence Theorem to the environment with incomplete information.
Surprisingly, we are able to identify an undemanding condition on
the payoff structure under which the formation process will again
converge even in this setting to the strongly efficient network in
equilibrium.

\subsection{Modeling Incomplete Information}

At the beginning of $t=1$, each agent only knows her own type and
holds the prior belief (that types are i.i.d. according to $H$) on
other agents' types.

Let $\mathcal{B}=\Delta(\Theta^{N})$ denote the set of possible beliefs
on the type vector. A (pure) \textbf{strategy} of agent $i$ is now
a mapping 
\begin{align*}
s_{i}:\mathcal{B}\times I-i\times Y\times\Omega\times Z\rightarrow\mathcal{A},
\end{align*}
with the same constraint $s_{i}(\cdot,\cdot,\cdot,0,0)\equiv0$. An
\textbf{equilibrium} is similarly defined as before, except for the
additional requirement that $i$ maximizes her expected discounted
total payoff given her belief at every period.

Let $B_{i}:Y\rightarrow\mathcal{B}$ denote agent $i$'s \textit{belief
updating function}, which is a mapping from the set of possible public
signals to the set of possible beliefs. We assume that it satisfies
the following properties: 
\begin{itemize}
\item {1.} $i$ knows her own type: regardless of $\sigma(t)$, she puts
probability $1$ on her true type. 
\item {2.} $i$ knows the type of any agent that she has been connected
to: if some ${\bf {g}}$ such that $ij\in{\bf {g}}$ has ever been
formed in $\sigma(t)$, then $i$ always puts probability $1$ on
$j$'s true type starting from period $t$. 
\item {3.} Agents use Bayesian updating whenever possible. We adopt the
convention that when Bayes rule does not apply, agents maintain the
same priors. 
\end{itemize}
We now define some concepts related to the payoff structure that will
be useful in constructing equilibrium strategies later. First we define
a \textit{partial equilibrium network} for a subset of agents.

\begin{defn}[Partial equilibrium network]Given $\bar{\theta}$, a
network ${\bf {g}}$ formed in $I'\subset I$ is a \textbf{partial
equilibrium network} for $I''\subset I'$ if (1) each agent in $I''$
gets a positive payoff from ${\bf {g}}$; (2) no agent in $I''$ can
increase her payoff by severing any of her links in ${\bf {g}}$.
\end{defn}

Given a subset of agents $I'$ and the associated type vector $\bar{\theta}_{I'}$,
consider a function $r:\Theta^{|I'|}\rightarrow G_{I'}$. We define
the following property for $r$:

\begin{defn}[Admissible function]We say that $r$ is an \textbf{admissible
function} for $I'$ if for every $\bar{\theta}_{I'}\in\Theta^{|I'|}$,
$r(\bar{\theta}_{I'})$ is a network such that (1) every non-singleton
agent in $r(\bar{\theta}_{I'})$ has a positive payoff; (2) there
exists a partial equilibrium network in $G_{I'}$ for the set of singleton
agents in $r(\bar{\theta}_{I'})$, denoted as $r'(\bar{\theta}_{I'})$.
We say that $I'$ is admissible if there exists an admissible function
$r$ for $I'$. \end{defn}

In a partial equilibrium network, no agent has incentive to unilaterally
sever any of her links; hence, the name ``partial equilibrium''.
Then (the existence of) an admissible function essentially characterizes
a particular type of agent subgroup: it maps every type vector in
the subgroup to a network that provides every connected agent a positive
payoff, and at the same time guarantees the existence of a partial
equilibrium network for the set of singleton agents. Intuitively,
the former network can be sustained in the long run when agents are
patient, and the latter can be used as a way to reward the future
singleton agents for their revelation of private information. We will
construct an equilibrium strategy profile following this argument
in the next section.

In many cases, the whole agent set $I$ is admissible. One of the
simplest scenarios is that for every type vector there exists a network
yielding a positive payoff for every agent, so that a partial equilibrium
network will not even be necessary because the set of singleton agents
will be empty. For instance, consider the connections model discussed
before, and consider the following two-type scenario: $\Theta=\{\alpha,\beta\}$,
$N=5$, $f(\alpha)>c>f(\beta)$, $(1+\delta)f(\beta)>c$, and $f(\alpha)+3f(\beta)>3c$.
First, note that when there exists at least one type $\alpha$ agent,
a partial equilibrium network for any number of type $\beta$ agents
is a star network with a type $\alpha$ agent as the center. Then
the following function $r$ is an admissible function for $I$:

\begin{center}
$r(\bar{\theta})=\left\{ \begin{array}{c}
\text{Star network with type \ensuremath{\alpha} center, if at least two agents are of type \ensuremath{\alpha}}\\
\text{Wheel network, otherwise.}
\end{array}\right.$ 
\par\end{center}

It is easy to verify that $r$ is indeed admissible for $I$. Moreover,
it is also straight forward to show that $r$ maps $\bar{\theta}$
to the strongly efficient network whenever the strongly efficient
network gives every agent a positive payoff. In this type of payoff
structure, the larger $N$ is, the more likely such a simple admissible
function for $I$ exists. When $N$ is large, even if $f(\alpha)$
and $f(\beta)$ are both small relative to $c$, a topology such as
a star or a wheel may still ensure a positive payoff for every agent.
In the case where $r$ maps a type vector to some unconnected network
(a network with singleton agents), a star or a wheel can be used as
partial equilibrium networks (in the case of a star, the singleton
agents in $r$ would be placed in the periphery).

\subsection{Construction of Equilibrium Strategies}

As with complete information, we explicitly construct a strategy profile
that will constitute an equilibrium when agents are sufficiently patient.
First we specify the associated monitoring structure, denoted $y_{ic}$.
Similarly to the case with complete information, $y_{ic}$ measures
the least informative monitoring structure needed for our network
convergence theorem.

To simplify the description of $y_{ic}$, we first introduce some
additional notations. For a given time period $t$ and a given subgroup
of agents $I'\subset I$, we denote as $I''_{1}({\bf {g}},t)$ the
subset of agents in $I'$ that failed to form/maintain a link in ${\bf {g}}$
when possible, and as $I''_{2}(t)$ the subset of agents in $I'$
that formed/maintained a link between $I'$ and $I\setminus I'$.
We let $\hat{{\bf {g}}}_{I'}$ denote the clique on $I'$. Finally,
we say that information is complete within $I'$ if every agent knows
the type of every other agent in $I'$, and that information is incomplete
within $I'$ otherwise. Denote these events in period $t$ as $O_{c}(I',t)$
and $O_{ic}(I',t)$. 
\begin{itemize}
\item {1.} $Y=\{X_{0},X_{1},T,E_{C},E_{P}\}\times2^{I}$. The subset of
$I$ in the second argument represents the subgroup of \textit{non-solitary
agents}, and $X_{0},X_{1},T,E_{C},E_{P}$ represent five phases with
respect to this subgroup: the \textit{experimentation phase with incomplete
information}, \textit{experimentation phase with complete information},
\textit{transition phase}, \textit{exploitation phase with cooperation}
and \textit{exploitation phase with punishment} correspondingly. We
will define and explain these concepts later. 
\item {2.} $y_{ic}(\sigma(0))=(X_{0},I)$. 
\item {3.} In period $t\geq1$, for every pair of agents $i,j\in I'\subset I$:

\begin{itemize}
\item {a.} If $y_{ic}(\sigma(t-1))=(X_{0},I')$: 
\begin{align*}
O_{c}(I'\setminus I''_{1}(\hat{{\bf {g}}}_{I'},t),t) & \rightarrow y_{ic}(\sigma(t))=(T,I'\setminus I''_{1}(\hat{{\bf {g}}}_{I'}),t)\\
O_{ic}(I'\setminus I''_{1}(\hat{{\bf {g}}}_{I'},t),t) & \rightarrow y_{ic}(\sigma(t))=(X_{0},I'\setminus I''_{1}(\hat{{\bf {g}}}_{I'}),t).
\end{align*}

\item {b.} If $y_{ic}(\sigma(t-1))=(X_{1},I')$: $y_{ic}(\sigma(t))=(T,I'\setminus(I''_{1}(r'(\bar{\theta}_{I'}),t)\cup I''_{2}(t)))$. 
\item {c.} If $y_{ic}(\sigma(t-1))=(T,I')$: $y_{ic}(\sigma(t))=(E_{C},I')$
if $r'(\bar{\theta}_{I'})$ has been the network topology within $I'$
(including no link between $I'$ and $I\setminus I'$) for a fixed
number of $J$ consecutive periods. Otherwise, $y_{ic}(\sigma(t))=(T,I'\setminus(I''_{1}(r'(\bar{\theta}_{I'}),t)\cup I''_{2}(t)))$. 
\item {d.} If $y_{ic}(\sigma(t-1))=(E_{C},I')$: 
\begin{align*}
I''_{1}(r(\bar{\theta}_{I'}),t)\cup I''_{2}(t)=\varnothing & \rightarrow y_{ic}(\sigma(t))=(E_{C},I')\\
I''_{1}(r(\bar{\theta}_{I'}),t)\cup I''_{2}(t)\neq\varnothing & \rightarrow y_{ic}(\sigma(t))=(E_{P},I')
\end{align*}

\item {e.} If $y_{ic}(\sigma(t-1))=(E_{P},I')$: if $y_{ic}=(E_{P},I')$
for a fixed number of $K$ consecutive periods, then $y_{ic}(\sigma(t))=(E_{C},I')$.
Otherwise, $y_{ic}(\sigma(t))=(E_{P},I')$. 
\end{itemize}
\end{itemize}
Essentially, the realization of $y_{ic}$ reveals publicly the current
phase of the game and whether agents in a certain subgroup $I'$ are
cooperating in every phase. The meaning of cooperation is phase-specific.
In the experimentation phase, agents are supposed to form and maintain
every link within $I'$ whenever possible, until information becomes
complete in $I'$ which brings the game into the transition phase.
Then cooperation among agents becomes forming the network $r'(\bar{\theta}_{I'})$
and keeping it for $J$ periods, with no link with $I\setminus I'$
at the same time. In these two phases, anyone who fails to cooperate
will be marked as a solitary agent. Afterwards, the game enters the
exploitation phase and the public signal works in the same way as
$y_{c}$ in the previous section, with $r(\bar{\theta}_{I'})$ as
the designated network.

Now we characterize the strategy profile based on $y_{ic}$, denoted
$\hat{s}_{ic}$. For every $i,j\in I$: 
\begin{itemize}
\item {1.} If $\max\{\omega_{ij},\zeta_{ij}\}=1$, the set of non-solitary
agents $I'$ is admissible and $i,j\in I'$, then $a_{ij}=1$ if any
of the following is true:

\begin{itemize}
\item {a.} $y_{ic}=(X_{0},I')$. 
\item {b.} $y_{ic}=(X_{1},I')$ and $ij\in r'(\bar{\theta}_{I'})$. 
\item {c.} $y_{ic}=(T,I')$ and $ij\in r'(\bar{\theta}_{I'})$. 
\item {d.} $y_{ic}=(E_{C},I')$ and $ij\in r(\bar{\theta}_{I'})$. 
\end{itemize}
\item {2.} $a_{ij}=0$ in all the other cases. 
\end{itemize}
In this strategy profile, agents in an admissible set cooperate as
much as they can according to the public signal. First, they reveal
their types by forming and maintaining links in the experimentation
phase until information becomes complete. In the transition phase
that follows, they form a partial equilibrium network $r'(\bar{\theta}_{I'})$
to provide positive payoffs to the singleton agents in network $r(\bar{\theta}_{I'})$,
the network that will persist in the long run. After $r'(\bar{\theta}_{I'})$
has existed for a specified length of time, the agents enter the exploitation
phase in which the formation process ultimately converges to $r(\bar{\theta}_{I'})$.
Agents who do not conform before the exploitation phase are categorized
as solitary agents and are left as singletons for ever, and those
who deviate during the exploitation phase only get temporary punishment.
Same as before, the exact deviating agent(s) cannot be identified,
so any punishment would be placed on pairs of agents rather than individual
ones.

\subsection{The Network Convergence Theorem with Incomplete Information}

In an environment with incomplete information, the Network Convergence
Theorem still holds in an admissible set of agents, but our constructed
equilibrium strategy profile leads to weak convergence only. The proof
below shows that when agents are sufficiently patient, (1) there exists
a length of punishment $K$ in the exploitation phase to ensure cooperation
and (2) there exists a length of reward $J$ in the transition phase
to ensure information revelation for the singleton agents in $r(\bar{\theta})$
in the experimentation phase.

\begin{thm} If $I$ is admissible, then for every admissible function
$r$ for $I$ there exists a cutoff $\bar{\gamma}\in(0,1)$ such that
if $\gamma\in[\bar{\gamma},1)$ and the true profile of types is $\bar{\theta}$,
then there exists an equilibrium in which the formation process converges
weakly to $r(\bar{\theta})$. \end{thm}

\begin{proof} Consider the monitoring structure $y_{ic}$ and the
strategy profile $\hat{s}_{ic}$. It suffices to show that for some
$J$ and $K$, there exists $\bar{\gamma}\in(0,1)$ such that for
all $\gamma\in[\bar{\gamma},1)$, $\hat{s}_{ic}$ is an equilibrium.
We need to check for sequential rationality given any possible formation
history. We proceed in the following order:

In the exploitation phase: since information is complete within $I''$,
sequential rationality is given by Theorem 1.

Given any formation history, for every solitary agent: given that
no other agent will agree to form a link with her, her subsequent
action will not affect her payoff, and hence sequential rationality
is satisfied.

Given any formation history such that the set of non-solitary agents
is not admissible, for every non-solitary agent: given that every
other agent is choosing action $0$, her subsequent action will not
affect her payoff, and hence sequential rationality is satisfied.

For every non-solitary agent in an admissible set $I'$ in the transition
phase: we need to discuss two cases. Following the proof of Theorem
1, we argue as follows: 
\begin{itemize}
\item {1.} For every singleton agent $i$ in $r(\bar{\theta}_{I'})$:
since $I'$ is admissible, $i$'s payoff in this phase is positive
(and bounded away from zero, from the assumptions that $\Theta$ is
finite and $I$ is finite) for $J$ periods in $r'(\bar{\theta}_{I'})$
if she follows the prescribed strategy. Also, her maximum expected
loss (negative payoff) before $r'(\bar{\theta}_{I'})$ is formed for
the first time and in the exploitation phase is bounded above regardless
of $\gamma$. Hence, given a sufficiently large $J$ and a sufficiently
large $\gamma$, $i$ does not have the incentive to deviate before
$r'(\bar{\theta}_{I'})$ is formed for the first time and become a
solitary agent (the payoff from which is bounded above regardless
of $\gamma$). After $r'(\bar{\theta}_{I'})$ is formed, since it
is a partial equilibrium network by assumption, there is no incentive
for $i$ to deviate and sever any of her links. 
\item {2.} For every non-singleton agent $j$ in $r(\bar{\theta}_{I'})$:
given $J$, $j$'s maximum expected loss in this phase and before
$r(\bar{\theta}_{I'})$ is formed in the exploitation phase is bounded
above regardless of $\gamma$. By the definition of $r$, $j$'s payoff
in $r(\bar{\theta}_{I'})$ is positive (and bounded away from zero,
by the same argument as above), which realizes every period after
$r(\bar{\theta}_{I'})$ is formed in the exploitation phase. Hence,
given a sufficiently large $\gamma$, $j$ does not have the incentive
to deviate in this phase and become a solitary agent. 
\end{itemize}
For every non-solitary agent in an admissible set $I'$ in the experimentation
phase: we need to discuss two cases. Following the proof of Theorem
1, we argue as follows: 
\begin{itemize}
\item {1.} For every singleton agent $i$ in $r(\bar{\theta}_{I'})$:
$i$'s maximum expected loss in this phase is bounded above regardless
of $\gamma$ since according to the strategy profile, information
becomes complete within finitely many periods almost surely. Then
with a similar argument to part (1) above, given a sufficiently large
$J$ and a sufficiently large $\gamma$, $i$ does not have the incentive
to deviate and become a solitary agent. 
\item {2.} For every non-singleton agent $j$ in $r(\bar{\theta}_{I'})$:
again, $j$'s maximum expected loss in this phase is bounded above
regardless of $\gamma$. Then with a similar argument to part (2)
above, given a sufficiently large $\gamma$, $j$ does not have the
incentive to deviate and become a solitary agent. 
\end{itemize}
This completes the proof. Note that there needs to be an upper bound
on $J$ in generic cases: given $\gamma$, the larger $J$ gets, the
less incentive a non-solitary and non-singleton agent in $r(\bar{\theta}_{I'})$
may have for following the prescribed strategy in the transition phase.
\end{proof}

The reason why we are not able to show strong convergence directly
is due to the incomplete information. Under complete information,
agents' beliefs on the type vector stay constant (on the true types)
over time despite the possibly changing public signals, which guarantees
unanimous knowledge on the payoff structure. Under incomplete information,
however, the beliefs can be heterogeneous and can evolve over time
according to the realization of public signals. The evolution of beliefs,
in turn, leads to each agent forming a belief on others' beliefs,
and hence it is difficult for them to agree on cooperation towards
one network topology. No matter how precise the public signal is (the
least precise being one constant signal, and the most precise being
equal to the formation history), this potential obstacle to coordination
exists as long as there is incomplete information on the type vector
among the agents.

Similar to Corollary 1, we obtain a result on sustaining a strongly
efficient network in equilibrium.

\begin{cor} Assume that $I$ is admissible. If $r(\bar{\theta})$
is strongly efficient for every $\bar{\theta}$, there exists $\bar{\gamma}\in(0,1)$
such that for all $\gamma\in[\bar{\gamma},1)$, there exists an equilibrium
where the formation process always converges to a strongly efficient
network weakly. \end{cor}

\subsection{Connections Model Revisited}

We have discussed how introducing incomplete information into the
strategic environment may have a considerable impact on the set of
sustainable networks. Constraining an agent's knowledge on the type
vector curbs her willingness to form costly links in anticipation
that such attempts may turn out to be futile. As a result, a sustainable
network under complete information may never emerge and persist under
incomplete information. Curiously, the impacts of complete and incomplete
information on the network formation process depend on agents' patience.
When agents are myopic, incomplete information can be a catalyst for
welfare improvement; when agents are very patient, it turns around
to become an impediment.

We use the previously discussed connections model to make the comparison.
Consider a connections model (with an arbitrary set of types) where
the payoff of agent $i$ from connecting to $j$ is $f(\theta_{j})$
(before spatial discount), and the expectation of payoff from a single
agent is larger than the link maintenance cost: $\mathbb{E}[f(\theta_{j})]>c$.
To avoid technical complications, we assume that in every network
the payoff of every connected agent is strictly positive. and that
the formation history is public knowledge: $y(\sigma(t))=\sigma(t)$.

Given a type vector $\bar{\theta}$ and a discount factor $\gamma$,
let $G_{c}^{W}(\bar{\theta},\gamma)$ ($G_{c}^{S}(\bar{\theta},\gamma)$)
denote the set of networks ${\bf {g}}$ such that there exists an
equilibrium where the formation process weakly (strongly) converges
to ${\bf {g}}$. Define $G_{ic}^{W}(\bar{\theta},\gamma)$ and $G_{ic}^{S}(\bar{\theta},\gamma)$
similarly.

\begin{prop} For every $\bar{\theta}$, there exists $\bar{\gamma},\underline{\gamma}\in(0,1)$
such that: 
\begin{itemize}
\item {a.} if $\gamma\in(0,\underline{\gamma})$, then $G_{c}^{S}(\bar{\theta},\gamma)\subset G_{c}^{W}(\bar{\theta},\gamma)\subset G_{ic}^{W}(\bar{\theta},\gamma)$ 
\item {b.} if $\gamma\in(\bar{\gamma},1)$, then $G_{ic}^{S}(\bar{\theta},\gamma)\subset G_{ic}^{W}(\bar{\theta},\gamma)\subset G_{c}^{W}(\bar{\theta},\gamma)=G_{c}^{S}(\bar{\theta},\gamma)$ 
\end{itemize}
\end{prop}

\begin{proof} $G_{c}^{S}(\bar{\theta},\gamma)\subset G_{c}^{W}(\bar{\theta},\gamma)$
and $G_{ic}^{S}(\bar{\theta},\gamma)\subset G_{ic}^{W}(\bar{\theta},\gamma)$
are clear from the definitions of strong and weak convergence. We
can apply Theorem 1 to show that $G_{c}^{W}(\bar{\theta},\gamma)=G_{c}^{S}(\bar{\theta},\gamma)=\{{\bf {g}}:u_{i}(\bar{\theta},{\bf {g}})>0\text{ for all \ensuremath{i}}\}$
when $\gamma$ is close to $1$. It immediately implies that $G_{ic}^{W}(\bar{\theta},\gamma)\subset G_{c}^{W}(\bar{\theta},\gamma)$
when $\gamma$ is close to $1$.

Now we prove that $G_{c}^{W}(\bar{\theta},\gamma)\subset G_{ic}^{W}(\bar{\theta},\gamma)$
when $\gamma$ is close to $0$. For every ${\bf {g}}$, suppose that
there exists an equilibrium $s_{c}$ under complete information where
the formation process converges weakly to ${\bf {g}}$. Consider the
following strategy profile $s_{ic}$ under incomplete information:
following any formation history $\sigma(t)$ which is on the equilibrium
path in $s_{c}$, in period $t+1$ if any link $ij$ is to be formed/maintained
according to $s_{c}$, then $a_{ij}=a_{ji}=1$; otherwise, $a_{ij}=a_{ji}=0$.
Following any formation history that is off the equilibrium path in
$s_{c}$, each agent chooses action $0$ thereafter. Clearly, this
strategy profile is an equilibrium following any formation history
off the equilibrium path in $s_{c}$. For every formation history
on the equilibrium path in $s_{c}^{y}$, it replicates the formation
process according to $s_{c}$. When $\gamma$ is sufficiently small,
each agent is indeed taking a best response. When link $ij$ is to
be formed/maintained according to $s_{c}$, if $i,j$ know each other's
type, the fact that they would have chosen to form/maintain the link
in $s_{c}$ implies that the current payoff from the link outweighs
the cost, so due to myopia $i,j$ will also form/maintain the link.
If $i,j$ do not know each other's type, the prescribed strategy profile
ensures that no Bayes' update occurs and thus their beliefs about
each other's type remains at the prior $H$. By the assumption that
$\mathbb{E}[f(\theta_{j})]>c$, $a_{ij}=a_{ji}=1$ is a best response.
When link $ij$ is to be severed/not formed according to $s_{c}$,
note that $a_{ij}=a_{ji}=0$ is a Nash equilibrium in a one-shot game
and implies mutual best response. Therefore, $s_{ic}$ is an equilibrium
under incomplete information and it produces a formation process identical
to the one under $s_{c}$. We can then conclude that when $\gamma$
is close to $0$, if ${\bf {g}}\in G_{c}^{W}(\bar{\theta},\gamma)$
then ${\bf {g}}\in G_{ic}^{W}(\bar{\theta},\gamma)$, which means
that $G_{c}^{W}(\bar{\theta},\gamma)\subset G_{ic}^{W}(\bar{\theta},\gamma)$.
\end{proof}

When agents are very patient, the set of sustainable networks is largest
when information is complete. To the contrary, when agents are myopic,
the set of sustainable networks can be larger when information is
incomplete. The reason is that when agents are myopic they have less
incentives to form links because they do not take account of future
benefits.

We conclude by showing that more networks are sustainable when agents
are ``more valuable'' whether or not information is complete or
agents are patient. Intuitively, when agents are more valuable, more
networks can be sustained in equilibrium because it is easier to provide
a larger set of agents with a positive payoff.

\begin{prop} Consider any two type vectors $\bar{\theta}$ and $\bar{\theta}'$.
If $f(\theta_{i})\geq f(\theta'_{i})$ for all $i$, then for every
$\gamma\in(0,1)$, $G_{m}^{n}(\bar{\theta},\gamma)\supset G_{m}^{n}(\bar{\theta}',\gamma)$,
$m=c,ic$, $n=S,W$. \end{prop}

\begin{proof} Consider a strategy profile $s$ (under either information
structure) such that when the type vector is $\bar{\theta}'$, $s$
is an equilibrium where the formation process converges (either strongly
or weakly) to a network ${\bf {g}}$. Modify it slightly in the following
way: whenever a link $ij$ is supposed to be severed/not formed according
to $s$, agents $i,j$ choose $a_{ij}=a_{ji}=0$. Now consider the
modified strategy profile $s'$ when the type vector is $\bar{\theta}$.
First, it is easy to see that it is indeed a best response to choose
$0$ when a link is supposed to be severed/not formed given that the
other agent is also choosing $0$. Secondly, note that given the same
formation history, the Bayes' update by any agent under $s$ and $s'$
is the same. As a result, whenever the prescribed action in $s'$
is $1$, it is again a best response due to the assumptions that $s$
is an equilibrium and that $f(\theta_{i})\geq f(\theta'_{i})$ for
all $i$. Therefore $s'$ is an equilibrium when the type vector is
$\bar{\theta}'$, and the formation process will converge in exactly
the same way. This completes the proof. \end{proof}

\section{Conclusion}

In this paper, we have studied the problem of dynamic network formation
by foresighted, heterogeneous agents under complete and incomplete
information. A large and growing literature has examined the network
formation process from various aspects, but the impact of agents'
foresight, hetorogeneity and incomplete information which make the
model truly realistic have not been studied. Existing works point
to a limited set of strongly efficient network topologies with homogeneous
agents, and the inability to sustain strongly efficient networks in
equilibrium. We question these results based on two grounds. On one
hand, the assumption of agent homogeneity is hard to justify in most
real-life applications. On the other hand, according to our characterization
of strongly efficient networks with heterogeneous agents and observations
in data collected from existing large networks, networks formed in
practical scenarios appear to be consistent to our predictions. Therefore,
we establish a dynamic network formation model to analyze the network
formation process and explain our findings.

In our model, agents meet randomly over time and voluntarily form
or sever links with each other. Link formation requires bilateral
consent but severance is unilateral. An agent's payoff in a single
period is determined by the network topology, her position in the
network, and the individual characteristics, also referred to as types,
of all agents she connects to (including herself). The agents are
foresighted in the sense that their final payoff is a discounted sum
of payoffs from each period. In every period, the agents observe the
set of their direct neighbors (the agents they link to) and a public
signal which is an indicator of the formation history.

We establish a Network Convergence Theorem under both complete and
incomplete information on the type vector, which characterizes the
set of sustainable networks in equilibrium for patient agents. Under
each environment, we show that a network can be sustained in equilibrium
as long as it provides each agent a positive payoff. As a corollary,
a strongly efficient network is sustainable when every agent's payoff
is positive, which presents a great contrast to the existing literature.
We argue that incomplete information is an important potential source
of inefficiency, which is corroborated by evaluating the lower bound
on agents' patience to sustain strongly efficient networks. Finally,
we use the connections model to fully characterize the set of strongly
efficient networks, whose topologies bear striking resemblance to
networks observed in data. This finding again confirms our theoretical
prediction that strongly efficient networks can be sustained in equilibrium.

We believe that many more problems regarding dynamic network formation
with foresightedness can be analyzed with the framework developed
in this paper. Questions that can be studied in future work include:
(1) how different stochastic meeting processes affect the level of
patience needed for sustaining efficient networks; (2) whether the
signal device can be generalized to transmit information only locally;
(3) in a connections model, how the spatial discount factor affects
the set of sustainable networks and the stability of efficient networks
for arbitrary time discounts.

\newpage{} \bibliographystyle{acm}
\bibliography{reference}

\end{document}